\documentclass[12pt]{article}
\usepackage{amsthm,amssymb,amsfonts,amsmath}
\usepackage[most]{tcolorbox}
\usepackage{bbm}
\usepackage{fancyhdr}
\usepackage{amsbsy,amsthm}
\usepackage{amscd}
\usepackage{latexsym}
\usepackage{graphicx} 

\usepackage{makecell}

\usepackage{pdfsync}
\usepackage{blkarray}
\usepackage{multirow}
\usepackage{hyperref}
\usepackage{color}

\usepackage{framed}
\usepackage{longtable}

\newcounter{constant} 
 \newcommand{\newconstant}[1]{\refstepcounter{constant}\label{#1}} 
 \newcommand{\useconstant}[1]{k_{\ref{#1}}}

\newcommand{\D}{\mathcal{D}_n}

\newcommand{\eps}{\varepsilon}

\newcommand{\calL}{\mathcal L}

\newcommand{\poly}{{\rm poly}}

\newcommand{\clh}{c_{\textsc{lh}}}

\newcommand{\bin}{{\rm Bin}}

\newtheorem{theorem}{Theorem}

\numberwithin{equation}{section}

\usepackage{todonotes}
\setuptodonotes{inline}

\usepackage{tikz}
\tcbuselibrary{skins}

\usetikzlibrary{calc, positioning, fit, shapes.misc}
\usetikzlibrary{shapes.geometric}

\usetikzlibrary{decorations.markings}

\tikzset{onethirdarrow/.style={
 decoration={markings,
 mark= at position 0.4 with {\arrow{#1}} ,
 },
 postaction={decorate}
 }
}

\tikzset{twothirdarrow/.style={
 decoration={markings,
 mark= at position \dfrac{1}{8}5 with {\arrow{#1}} ,
 },
 postaction={decorate}
 }
}

\usepackage{enumitem}
\setlist{nosep, leftmargin=.6cm}
\usepackage{bm}
\numberwithin{equation}{section}

\usepackage{tikz}
\usepackage{graphicx}
\usepackage{float}

\usepackage{breakcites}
\hypersetup{colorlinks=true, citecolor=blue, linkcolor=blue}

\usepackage{algorithmic}
\usepackage[linesnumbered]{algorithm2e}

\renewcommand{\ge}{\geqslant}
\renewcommand{\le}{\leqslant}
\renewcommand{\geq}{\geqslant}
\renewcommand{\leq}{\leqslant}

\newcommand{\E}[1]{\mathbf E\left[#1\right]}

\renewcommand{\P}[1]{\mathbf{P}_p\left(#1\right)}

\newcommand{\Exp}[1]{\exp\left(#1\right)}
\renewcommand{\O}[1]{O\left(#1\right)}

\renewcommand{\o}[1]{o\left(#1\right)}
\newcommand{\Log}[1]{\log\left(#1\right)}

\newcommand{\ep}{\varepsilon}

\newcommand{\p}{\bm p}
\newcommand{\q}{\bm q}

\renewcommand{\i}{\bm i}
\renewcommand{\j}{\bm j}

\newcommand{\I}[1]{\mathbf 1_{\left\{#1\right\}}}

\newtheorem{tm}{Theorem}[section]

\newtheorem{lm}[tm]{Lemma}

\newtheorem{cl}[tm]{Claim}

\theoremstyle{definition}
\newtheorem{df}[tm]{Definition}

\newcommand{\bR}[1]{\left(#1\right)}
\newcommand{\bC}[1]{\left\{#1\right\}}

\def\abs#1{\big\lvert#1\big\rvert}

\newcommand{\floor}[1]{\left\lfloor#1\right\rfloor}
\newcommand{\ceil}[1]{\left\lceil#1\right\rceil}

\def\dfrac#1#2{
\hbox{\Large$\textstyle\frac{#1}{#2}$}}

\usepackage[margin=1in]{geometry}

\usepackage[width=.9\textwidth]{caption}

\let\originalleft\left
\let\originalright\right
\renewcommand{\left}{\mathopen{}\mathclose\bgroup\originalleft}
\renewcommand{\right}{\aftergroup\egroup\originalright}

\begin{document}

\title{
\Large{\bf Finding a Nash equilibrium of a random win-lose game 
in expected polynomial time}}
\author{
\normalsize	 Andrea Collevecchio\thanks{School of Mathematics, Monash University, Clayton, VIC, Australia:
{\tt andrea.collevecchio@monash.edu}.}
\and
\normalsize	 G\'abor Lugosi\thanks{Department of Economics and Business, 
Pompeu Fabra University;
Barcelona School of Economics;
ICREA, Pg. Lluis Companys 23, 
08010 Barcelona, Spain:
{\tt gabor.lugosi@gmail.com}.}
\and
\normalsize	 Adrian Vetta\thanks{
Department of Mathematics and Statistics,
and School of Computer Science,
McGill University, Montreal,
QC, Canada:
{\tt adrian.vetta@mcgill.ca}.}
\and
\normalsize	Rui-Ray Zhang\thanks{
{\tt rui.ray.zhang@hotmail.com}.
Partially supported by an Alf van der Poorten Travelling Fellowship.}
}
\date{}

\maketitle

\begin{abstract}
A long-standing open problem in algorithmic game theory asks whether or not there is a polynomial
time algorithm to compute a Nash equilibrium in a random bimatrix game.
We study random win-lose games, where the entries of the $n\times n$ payoff matrices are independent and identically distributed (i.i.d.) Bernoulli random variables with parameter $p=p(n)$.
We prove that, for nearly all values of the parameter $p=p(n)$, 
there is an expected polynomial-time algorithm to find a Nash equilibrium in a random win-lose game. 
More precisely, if $p\sim cn^{-a}$ for some parameters
$a,c\ge 0$, then there is an expected polynomial-time algorithm
whenever $a\not\in \{1/2, 1\}$. 
In addition, if $a = 1/2$ there is an efficient algorithm if either $c \le  e^{-52} 2^{-8} $ or $c\ge 0.977$. If $a=1$, then there is an expected polynomial-time algorithm if either $c\le 0.3849$ or $c\ge \log^9 n$.
\end{abstract}


\section{Introduction}\label{sec:intro}
The fundamental question in the development of algorithmic game theory has been the complexity of computing a Nash equilibrium; see~\cite{Pap01}.
A mixed Nash equilibrium (MNE) in a normal-form game is a profile in the set of randomized strategies such that no player, considered individually, has an incentive to change strategy. 
The existence of Nash equilibria is well-settled: Von Neumann's minimax theorem~\cite{VN28} implies the existence of a MNE in a zero-sum, two-player game, and Nash~\cite{Nash51} later famously proved that every finite game has at least one MNE. In other words, the players' strategic interactions can always be resolved into a stable outcome. 
This correspondence between equilibria and stability has since played a crucial role in game theory and economics more generally.

\subsection{Computing a Nash equilibrium}
A mixed Nash equilibrium always exists in a finite game.
However, finding a MNE is computationally
intractable, even for games with only two players. 
In celebrated work, Daskalakis, Goldberg and Papadimitriou~\cite{DGP09} proved the problem to be PPAD-complete for $4$-player games.
Building upon this, 
Chen, Deng and Teng 
\cite{CDT09} obtained the same result 
for $2$-player (bimatrix) games. Specifically, given payoff matrices 
$\bm A=(a(i,j))_{i,j \in [n]}$ for the row player 
and $\bm B=(b(i,j))_{i,j \in [n]}$ for the column player, 
where $[n]:=\{1,2, \ldots, n\}$,
computing 
(or even, approximately computing) a MNE is 
 PPAD-complete. This immediately prompts the question as to whether special classes of bimatrix game are easier to solve. 
Surprisingly, the answer is \emph{no} even for the following very simple classes:\\
\indent$\bullet$ {\em Sparse Games}, where the number of non-zero entries in every row and column is at most a constant \cite{CDT06}. \\
\indent$\bullet$ {\em Win-Lose Games}, where every payoff entry is either zero or one \cite{AKV05}.\\
\indent$\bullet$ {\em Constant Rank Games}, where the sum of the two payoff matrices, $\bm A+ \bm B$, has a constant rank of at least $3$ \cite{Meh14}.

Indeed, the picture is bleak in terms of classes of two-player games that are solvable in polynomial time. As stated, rank-zero games (i.e., zero-sum games) can be solved in polynomial time. Furthermore, Adsul et al.~\cite{AGM11} showed that rank-one games are also polytime solvable. 
Beyond that, positive results are known only in extremely restrictive settings; for example, MNE can be computed efficiently in win-lose games
when an associated digraph representation (see Section \ref{sec:graph}) is planar \cite{AOV07} or has out-degree at most two \cite{CLR06}.

This observation motivates our research. In addition to games of rank at most one, are there any non-trivial classes of
bimatrix games that are solvable in polynomial time?
Given this objective, we study the class of 
{\em random games}, where the entries of the payoff matrices are drawn independently from a distribution.
The investigation of random games stems from the supposition
that these games may have a common structure
allowing for the design of efficient algorithms. 
An illustration of this comes from the work of
B\'{a}r\'{a}ny, Vempala, and Vetta~\cite{BVV07} 
who examined random bimatrix games with entries drawn independently from a normal (or uniform) distribution. 
Their key structural insight was that, with probability $1-O\left(\log^{-1} n\right)$, such games
contain Nash equilibria with supports of cardinality two.
Accordingly, they showed that a simple support enumeration algorithm will find a Nash equilibrium in the random game in $O(n^3 \log \log n)$ time, with high probability.
But, unfortunately, the proven failure probability of their algorithm is far too large to produce an expected polynomial-time algorithm.
In particular, if the algorithm fails to find a small support Nash equilibrium, it is forced to use a generic
all-purpose algorithm to find an equilibrium.
The fastest and best-known such algorithm is the  
Lemke-Howson algorithm \cite{LH64}, which, for any $n\times n$ bimatrix game, finds a Nash equilibrium in time
$O(\clh^n)$, where $\clh$ is the \emph{Lemke-Howson constant}.
It is known~\cite[Corollary 2.13]{von2002computing} that an upper bound on this constant is
\begin{equation}\label{eq:clh}
\clh \le 2.598.
\end{equation}
Further, there are games where the Lemke-Howson algorithm does take exponential time~\cite{savani2006hard}.

Consequently, B\'{a}r\'{a}ny, Vempala and Vetta~\cite{BVV07} posed the question of finding an expected polynomial-time algorithm for random games, either via a new algorithm or by providing a better analysis 
of their algorithm. This problem has remained unsolved for 20 years, even for random win-lose games, where the bimatrix entries
are drawn independently from a Bernoulli distribution.
In this paper, our contribution is an expected polynomial-time algorithm for random win-lose games, for nearly all choices of the Bernoulli parameter $p=p(n)$.

\subsection{Overview and results}
We study random win-lose games, in which the entries of the $n\times n$ payoff matrices $\bm A$ and $\bm B$ 
are independent Bernoulli random variables with parameter $p=p(n)$.
In other words, each entry is $1$ with probability $p$, and $0$ with probability $1-p$, independently.
We prove that for nearly all values of the parameter $p$, there is an expected polynomial-time algorithm that computes a Nash equilibrium.
In particular, we show that for $p\sim cn^{-a}$
there is an expected polynomial-time algorithm 
for all parameters
$a \in (0,1/2) \cup (1/2,1) \cup (1, \infty),$ and all $c\ge 0$.
This leaves the special cases of $a=1/2$ and $a=1$.
If $a = 1/2$ we prove that the algorithm finds a MNE in polynomial time if either $c \le  e^{-52} 2^{-8} $ or $c>0.977$.
Furthermore, if $a=1$ there is an expected polynomial-time algorithm for $c\le 0.3849$ and $c\ge \log^9 n$.


The key observation in proving these results is that
win-lose games have a directed bipartite graph representation;
see Section~\ref{sec:graph} for details.
Further, we show that various combinatorial structures in these graphs correspond to Nash equilibria.
This allows us to design a collection of algorithms that search for
these combinatorial structures.
These algorithms have different ranges with respect to
the parameter $p$ in which they succeed in finding these structures, and hence a Nash equilibrium, in expected polynomial time. 
Table~\ref{tab:results} summarizes our results.

\begin{table}[h]
\centering
\begin{tabular}{c|c|c}
& Constraints on $p$ & Search Strategy \\ \hline\hline
Regime 0: Section \ref{sec:reg0} & $p \le 0.3849 n^{-1}$ & \sc{PNE} \\ \hline
Regime 1: Section \ref{sec:reg1} 
& $(\log n)^{9} n^{-1} 
\le p
\le
e^{-4/3} n^{-2/3} $
& cycles of $p$-dependent length \\ \hline
The Middle Gaps: Section \ref{sec:gap} & 
\makecell{
$e^{-4/3} n^{-2/3} \le p \le 2^9 n^{-2/3}$ \\
$e^{3/2} n^{-3/4 }
\le p \le
e^3 n^{-3/4}$}
& $4$-cycles and $6$-cycles \\ \hline
Regime 2: Section \ref{sec:reg2} & 
$2^9 n^{-2/3} \le p \le  e^{-52} 2^{-8}  n^{-1/2}$
& $4$-cycles with localized pair \\ \hline
Regime 0: Section \ref{sec:reg0} 
& $p \ge 0.977 n^{-1/2}$ & \sc{PNE} \\ \hline
\end{tabular}
\caption{Summary of Results.}
\label{tab:results}
\end{table}

In Section~\ref{sec:generic}, we present a generic 3-step algorithm for finding a Nash equilibrium. In the first step, it searches for pure Nash equilibria. If that fails, in the second step, it searches for mixed Nash equilibria with supports of a specified cardinality $\ell$. If that also fails, then in the third step, it runs the Lemke-Howson algorithm. The choice of $\ell$ is critical and essentially determines what combinatorial structures the algorithm is looking for. 

For the choice of $\ell=1$, the algorithm just searches for pure Nash equilibrium before executing the Lemke-Howson algorithm. We prove, in Section~\ref{sec:reg0}, that this trivial algorithm is efficient when $p$ is either extremely small ($p \le 0.3849 n^{-1}$) or is very large ($p \ge 0.977 n^{-1/2}$).
This search for PNE we call ``Regime 0".

The search for a mixed Nash equilibrium is non-trivial.
The main combinatorial structure we search for is a special type
of directed cycle that we call a {\em stable cycle}; see Section~\ref{sec:graph}.
When $p$ is moderately small ($\log^9 n \, n^{-1} 
\le p \le e^{-4/3} n^{-2/3}$), we search for long stable cycles
whose length is dependent on $p$. Thus, we run the generic
algorithm where $\ell$ is dependent on $p$, which we dub ``Regime 1". The resultant performance analysis, given in Section~\ref{sec:reg1}, then involves the study of disjoint cycle packings.

When $p$ is moderately large ($2^9 n^{-2/3} \le p \le  e^{-52} 2^{-8}  n^{-1/2}$) we search for stable cycles
of length $\ell=4$. To prove this is efficient, we again study
disjoint cycle packings, but with a caveat. We insist that each cycle
in packing contains a pair of vertices that are close in some ordering (a {\em localized pair}). This ``Regime 2" is covered in Section~\ref{sec:reg2}.

Notice there is a small gap between Regimes 1 and 2, namely
$e^{-4/3} n^{-2/3} \le p \le 2^9 n^{-2/3}$. Moreover, there is a
small interval within ``Regime 1", specifically $e^{3/2} n^{-3/4 }
\le p \le e^3 n^{-3/4}$, where our prior analysis fails.
To rectify this, in Section~\ref{sec:gap}, we apply a
probabilistic method called sprinkling. We then show how searching for both $4$-cycles and $6$-cycles will give an expected polynomial-time algorithm for these intervals.

Finally, the reader may observe that there are two extremely small intervals for $p$ that remain open.
First when $p\sim c/n$ for a constant $c\ge 0.3849$ or $pn \to \infty$ but $p=O(\log^9n/n)$, and second 
when $p \sim cn^{-1/2}$ for $c \in ( e^{-52} 2^{-8} , 0.977)$.
These two cases are intriguing challenges for future research.
%
%
%
%
%
%
%
%
%
%
%
%

\subsection{Related work on random games}

Random games have been studied as far back as the works
of Goldman~\cite{Gold57} and Goldberg et al.~\cite{GGN68}.
The literature focuses predominantly on pure Nash equilibria (PNE), that is, profiles of pure strategies such that none of the players, picked individually, 
has an incentive to deviate. Furthermore, much work has been devoted to describing the expected number of PNE and how to find them through adaptive/learning processes, such as best response dynamics or better response dynamics. If we consider a game with i.i.d. random payoffs from a continuous distribution, then when the number of players (or the number of actions) grows to infinity, the number of PNE is asymptotically Poisson(1); see, e.g., Rinott and Scarsini \cite{R2000}.
The number of PNE in the case of two-actions, a large number of players, and i.i.d. payoff with atoms was analysed in 
Amiet et al.~\cite{ACSZ19}. The latter provides a central limit theorem for the number of PNE and shows that for certain regimes, PNE can be found by simple learning dynamics, with high probability. 
The number of mixed Nash equilibria has also been studied in detail; see, for example, the works of McLennan~\cite{McL05}, and McLennan and Berg~\cite{MB05}.

\section{Nash equilibria in win-lose games}\label{sec:generic}

In this section, we formally define win-lose bimatrix games and
the concept of a Nash equilibrium. We then present a generic
algorithm for finding a Nash equilibrium in a game.

\subsection{Preliminaries}
Consider a game with two players, where 
the number of available strategies for each player is $n$. 
We refer to the two players as the row and the column player
and denote their payoff matrices by $\bm A = (a(i, j))_{i, j \in [n]}$ 
and $\bm B = (b(i, j))_{i, j \in [n]}$ respectively, recalling $[n]=\{1, 2, \ldots, n \}$.
This is known in the literature as a {\em bimatrix game}.
In this paper, we focus on {\em win-lose} bimatrix games where every entry of
$\bm A$ and $\bm B$ is in $\{0,1\}$. Note,
for simplicity, we assume both players have the same number of strategies, 
though our techniques may be generalized to study games with a different number 
of pure strategies for each player.

A \emph{mixed strategy} for a player is 
a probability distribution over the set of pure strategies,
and is represented by a probability vector $\p = (p_1, p_2, \ldots, p_n)$, 
where $p_i$ denotes the probability that the player chooses the $i$th pure strategy
for $i \in [n]$; that is the pure strategy row $i$ for the row player or the pure strategy column $i$ for the column player. 
The \emph{support} of $\p$ is the set of indices $i\in [n]$ with $p_i\neq 0$.
For a mixed strategy pair $( \p, \q )$, the payoff to the row player 
is the expected value of a random variable that takes the value $a(i,j)$ with probability $p_i q_j$. 
In other words, the payoff to the row player is $\p^\intercal \bm A \q$,
where $\p^\intercal$ denotes the transpose of the vector $\p$.
Similarly, the payoff of the column player is $\p^\intercal \bm B \q$. A stable solution in the resultant game, namely a {\em Nash equilibrium}, is defined as follows.
\begin{tcolorbox}
\begin{df}
A pair of mixed strategies $(\p, \q)$ is a \emph{Nash equilibrium} if and only if\\
 $\bullet$ for every mixed strategy $\p_0$ of the row player,
$\p_0^\intercal \bm A \q \le \p^\intercal \bm B \q$, and\\
$\bullet$ for every mixed strategy $\q_0$ of the column player,
$\p^\intercal \bm A \q_0 \le \p^\intercal \bm B \q$.
\end{df}
\end{tcolorbox}

\subsection{A generic algorithm}
\label{sec:generic}
The basic idea of our approach is to look for sparsely supported Nash equilibria. 
For example, suppose a game has a Nash equilibrium whose supports have cardinality $\ell$ for both players, namely $\ell \times \ell$ supports. Then such an equilibrium can be found in time $O(n^{2\ell} \poly(n))$, where $\poly(n)$ stands for a polynomial in $n$, by the following reasoning. Enumerate over
all $\binom{n}{\ell}^2
= \O{ n^{2\ell} }$ supports of size $\ell\times \ell$ and test if any induce a 
Nash equilibrium. For each pair of supports, this can be done in time $\poly(n)$ via a standard linear program. (Alternatively, for our purposes, it will suffice to use combinatorial algorithms
that search for Nash equilibria corresponding to the aforementioned stable cycles.)
If this approach fails, then we simply run the Lemke-Howson algorithm to find a Nash equilibrium with large supports.
This idea leads to the following generic algorithm for finding a Nash equilibrium.

%
\begin{tcolorbox}[title={The Generic Algorithm with Parameter $\ell$}]
\begin{enumerate}
\item Search the game for PNE; if one is found, output that equilibrium.
\item
If no PNE is found in Step 1, then exhaustively search for MNE with $\ell\times \ell$ supports. 
\item[]
If such a MNE is found, then output that equilibrium.
\item
If no MNE is found in Step 2, then run the Lemke-Howson algorithm and output the equilibrium obtained.
\end{enumerate}
\end{tcolorbox}

The expected running time of the generic algorithm is at most
\begin{align}
&O(n^2) 
+ \O{\binom{n}{\ell}^2 \poly(n) } \P{\mbox{no PNE is found in Step 1}} \notag\\
&\qquad\qquad\qquad\qquad
\qquad\qquad\qquad\qquad
+ \O{ \clh^n } \P{\mbox{no MNE is found in Step 2}},
\label{time0}
\end{align}
where $\clh \le 2.598$ is the Lemke-Howson constant defined in \eqref{eq:clh}.
To understand $O(n^2)$ running time for Step 1, consider when
a row-column pair $\{r_i, c_j\}$ is a PNE
in a bimatrix win-lose game.
If $a(i,j)=b(i,j)=1$ then the pair is a PNE.
If $a(i,j)=1$ and $b(i,j)=0$ then the pair is a PNE if and only
if the $i$th row of $\bm B$ is all zero; symmetrically, 
if $a(i,j)=0$ and $b(i,j)=1$ then the pair is a PNE if and only
if the $j$th column of $\bm A$ is all zero.
If $a(i,j)=b(i,j)=1$ then the pair is a PNE.
If $a(i,j)=b(i,j)=0$ then the pair is a PNE if and only
if the $j$th column of $\bm A$ and the $i$th row of $\bm B$ are both all zero. The expected running times of Step 2 and Step 3 follow the prior discussions.

For example, suppose for a moment that $\ell$ is a fixed constant.
Then in order to show that for a certain value of $p$, the generic algorithm runs in expected polynomial time,
it suffices to show that the probability that the game has no $\ell \times \ell$ Nash equilibrium is 
bounded by 
$\clh^{-n}$ 
for sufficiently large $n$. To this end, we consider various regimes of the parameter $p=p(n)$.

We also remark that any advance in improving the upper bound on $\clh$
will automatically widen the intervals of success for the Regimes shown in Table~\ref{tab:results}.

\section{Regime $0$: pure Nash equilibria}
\label{sec:reg0}

We are now ready to give an algorithm for 
a random win-lose game, where the probability a payoff entry is $1$ is $p=p(n)\in [0,1]$. Recall, to achieve an expected polynomial running time, we will apply the generic algorithm, defined in Section~\ref{sec:generic}, with a different parameter $\ell$, depending upon the value of $p$. 
In this section, we begin with the two simplest cases where $p$ is
either very small or large (Regime $0$). In both these cases, we prove that a random win-lose game has a pure strategy Nash equilibrium with
sufficiently large probability to guarantee that 
the generic algorithm with parameter $\ell=1$ does indeed find a Nash equilibrium in expected polynomial time. This is formalized by the following theorem.
\begin{tcolorbox}
\begin{theorem} 
\label{thm:pure}
Suppose that 
\begin{align*}
\mbox{either } \quad \inf_n \, p \sqrt{n} \ge \sqrt{ \log \clh },
\quad\text{ or }\quad
\sup_n \, p n \le 1/\clh.
\end{align*}
Then, the generic algorithm 
of Section~\ref{sec:generic} with $\ell=1$ finds a Nash equilibrium in expected time $O(n^2)$.
\end{theorem}
\end{tcolorbox}

\begin{proof}
Recall that the payoff matrices $\bm A$ and $\bm B$ are independent $n \times n$ Bernoulli random matrices
with parameter $p$.
If $a(i,j)=b(i,j)=1$
for some $i,j\in [n]$, then the pair $(i,j)$ of actions forms a PNE, as the corresponding payoff is $(1,1)$. 
The probability that no such payoff exists equals
\begin{align}
\P{ \mbox{no PNE in } (\bm A, \bm B) }
= (1 - p^2)^{n^2} \le e^{-p^2n^2}.
\label{pne}
\end{align}
This is bounded by $\clh^{-n}$ whenever 
\begin{align*}
p \ge \sqrt{\dfrac{\log(\clh)}{n}} \ge \dfrac{0.977}{\sqrt{n}}.
\end{align*}

Another simple instance of a PNE is when
there is an all-$0$ row in the matrix $\bm B$ \emph{and} there is
an all-$0$ column in the matrix $\bm A$. 
More precisely, suppose for
some $i,j\in [n]$, both 
$a(i,j') = 0$
and 
$b(i',j)=0$
for all $i',j'\in [n]$. Then
the pair $(i,j)$ forms a PNE.  
Since 
\begin{align*}
\P{ \mbox{$\bm B$ has an all-$0$ row} }
&= 1 - \P{ \mbox{all rows of $B$ contain some } 1 } \\
&= 1 - \P{ \mbox{row $1$ has some } 1 }^n \\
&= 1 - (1 - \P{ \mbox{row $1$ has only $0$s} } )^n \\
&= 1 - (1 - (1-p)^n )^n~,
\end{align*}
the probability that no such pair $(i,j)$ exists equals
\begin{eqnarray*}
 1- ( 1 - (1 - (1-p)^n )^n )^2
\le
 1 - \left( 1 - t^n \right)^2,
\end{eqnarray*}
where 
\begin{align*}
t = 1 - \Exp{ - \frac{np}{1-p} },
\end{align*}
since $1-x \ge e^{-\frac{x}{1-x}}$ for all $x < 1$.
As $(1-x)^2 \ge 1-2x$ for all $x \in \mathbb R$,
we have
\begin{eqnarray*}
 1- ( 1 - (1 - (1-p)^n )^n )^2
\le 2t^n
= 2 \left(1 - \Exp{ - \frac{np}{1-p} } \right)^n.
\end{eqnarray*}

Hence
if $p\le \alpha/n$ for some constant $\alpha >0$, then using the fact that $1-e^{-x} \le x$ for all $x>0$, we have
\[
 1- ( 1 - (1 - (1-p)^n )^n )^2 \le 2 \left(\frac{\alpha}{1-\alpha/n} \right)^n 
 = 
2 \alpha^n \left(1- \dfrac\alpha n \right)^{-n} 
\le 3e^\alpha \alpha^n,
\]
as $(1- \alpha /n )^{-n}$ converges to $e^\alpha$ as $n\to\infty$. Hence, for all sufficiently large $n$, we have $(1- \alpha/ n)^{-n} \le 3 e^\alpha/2$.
Now $3e^\alpha \alpha^n$ is bounded by a constant multiple of $\clh^{-n}$ whenever $\alpha \le 1/\clh$.
\end{proof}

\section{Combinatorial structures and MNE}\label{sec:graph}
The results in Theorem \ref{thm:pure} described in the previous section do not extend to other ranges of $p$. 
Our bounds suggest that in other regimes, the probability 
that a PNE exists is not large enough to guarantee a polynomial expected running time for the algorithm. Instead, we look for MNE.
To do this, in this section, we explain how certain MNE correspond
to specific combinatorial structures in an auxiliary directed graph
associated with the bimatrix game $(\bm A, \bm B)$.
Specifically, these MNE have supports that correspond to a cycle in the graph. Moreover, these cycles require additional global properties with respect to the graph to ensure best response stability, so we call them {\em stable cycles}.

\subsection{A directed bipartite graph representation}
As stated, it is useful to represent a win-lose game by a directed bipartite graph.
Recalling that $\bm A$ and $\bm B$ are two $n \times n$ Bernoulli random matrices,
we define a random directed bipartite graph $\D = \mathcal D( \bm A, \bm B )$ as follows.
The directed bipartite graph $\D$ has 
a vertex bipartition, $R$ and $C$,
both containing $n$ vertices.
We call $R$ and $C$ the left and right parts, respectively.
There is a vertex $r_i \in R$ for each row $i \in [n]$ 
and a vertex $c_j \in C$ for each column $j \in [n]$. 
The digraph $\D$ contains arc $(r_i, c_j)$ if and only if $ b(i, j) = 1$; 
similarly, $\D$ contains arc
$(c_j, r_i)$ if and only if $ a(i, j) = 1$. 
We have that 
$
\mathbf{P}_p( (r_i, c_j)) \in E) = \mathbf{P}_p( (c_j, r_i)) \in E)=p$, and each arc is included independently of the others. It is worth noting that if a pair of pure strategies $\{r_i, c_j\}$ is contained in a directed cycle of length two, then $(i, j)$ is a PNE.

\subsection{Stable cycles}

In our quest for finding mixed Nash equilibria, we start with a simple observation that is
crucial for the rest of the paper: A certain simple configuration 
corresponds to a Nash equilibrium. This configuration, which we call a \emph{stable cycle}, is defined as follows.

\begin{tcolorbox}
\begin{df}\label{def:unst}
Fix $\ell \in [n]$.
A {\em stable $2\ell$-cycle} in $\D$, say $H$,
is a directed cycle in $\D$ on $2\ell$ vertices 
such that no other vertex 
in $V(\D)\setminus V(H)$ is a common out-neighbour of any two vertices of cycle $H$.
In other words, 
let $A = V(H) \cap R$ and $B = V(H) \cap C$;
for each $i\in R \setminus A$, there is at most one $j\in B$
such that $(c_j, r_i) \in \D$ \emph{and} for each $j\in C\setminus B$ there is at most one $i\in A$
such that $(r_i, c_j)\in \D$.\\
A $2\ell$-cycle that is not stable is called \emph{unstable}.
\end{df}
\end{tcolorbox}
%
%
%
%
%
%

Next, we show that every stable cycle corresponds to a mixed Nash equilibrium.
In the remaining part of the paper, we analyse the existence of stable cycles of length $2 \ell$.
\begin{tcolorbox}
\begin{lm}\label{lm:stableNE}
For every stable $2\ell$-cycle, there exists a mixed strategy Nash equilibrium $(\p,\q)$, whose support
corresponds to the vertices on the left part and right part of the cycle, respectively.
\end{lm}
\end{tcolorbox}
\begin{proof} 
Let $A \subset R$ and $B \subset C$ be such that $A \cup B$ is the vertex set of a stable $2\ell$-cycle in $\D$.
Consider the subgame induced by the rows in $A$ and the columns in $B$. Let this subgame have a Nash equilibrium, in which 
the row player selects the probability distribution $\p$ 
and the column player selects the probability distribution $\q$. 
Such an equilibrium exists by Nash's theorem. We claim
that $( \p, \q )$ is a Nash equilibrium for the whole game
(where $\p$ and $\q$ are extended to place zero probability on
the strategies in $R\setminus A$ and $C\setminus B$, respectively). To see this, let $r$ be a column in $R\setminus A$. As the cycle $A\cup B$ is stable, $r$ is the out-neighbour of at most one vertex in $B$. If $r$ is not the out-neighbour of any vertices in $B$, then the row gives payoff zero against $\q$. So we may assume $r$ is the out-neighbour of exactly one vertex $c$ in $B$. But $c$ has an out-neighbour $r'$ on the stable cycle $A\cup B$. Thus the payoff of $r'$ against $\q$
is at least as large as the payoff of $r$ against $\q$ (possibly $r'$ has a greater response payoff if it is the out-neighbour of more than one vertex on the cycle). Thus the pure strategy $r$ is not a better strategy than $\p$ in response to $\q$.
So $\p$ is the best response to $\q$ in the whole game.
A similar argument shows $\q$ is the best response to $\p$ in the whole game.
\end{proof}


Ideally, if $\ell$ is small, then searching for stable $2\ell$-cycles can be done quickly. However, as we will see in Section~\ref{sec:reg1}, we will often be forced to search for
relatively large stable $2\ell$-cycles. Even more critically,
if our search fails, then we must run the exponential time
Lemke-Howson algorithm. As such, the following lemma will be useful
in bounding the overall expected running time~\eqref{time0} of the generic algorithm. 

\begin{tcolorbox}
\begin{lm}
\label{Pfail}
Let integer $\ell \ge 1$.
Then for any $0 < \beta < 1/\ell$,
\begin{align*}
\P{ \mathcal{D}_n \text{ has no 
stable $2\ell$-cycle} } 
&\le 
\P{ \mathcal{D}_n \text{ has less than $\beta n$ disjoint $2\ell$-cycles} } \\
&\qquad
+ \P{ \mathcal{D}_n \text{ has at least $\beta n $ disjoint unstable $2\ell$-cycles} }.
\end{align*}
\end{lm}
\end{tcolorbox}
\begin{proof}
By the law of total probability, we have
\begin{align*}
&\P{ \mathcal{D}_n \text{ has no 
stable $2\ell$-cycle} } \\
&= 
\P{ \mathcal{D}_n \text{ has no 
stable $2\ell$-cycle and $\le \beta n$ disjoint $2\ell$-cycles} } \\
&\qquad+
\P{ \mathcal{D}_n \text{ has no
stable $2\ell$-cycle and $\ge \beta n$ disjoint $2\ell$-cycles} } \\
&\le
\P{ \mathcal{D}_n \text{ has $\le \beta n$ disjoint $2\ell$-cycles} } \\
&\qquad+
\P{ \mathcal{D}_n \text{ has no
stable $2\ell$-cycle and $\ge \beta n$ disjoint $2\ell$-cycles} }\\
&\le
\P{ \mathcal{D}_n \text{ has $\le \beta n$ disjoint $2\ell$-cycles} }
+ \P{ \mathcal{D}_n \text{ has $\ge \beta n$ disjoint unstable $2\ell$-cycles} }.
\end{align*}
This completes the proof.
\end{proof}

\section{Regime 1: disjoint stable $2\ell$-cycles}
\label{sec:reg1}

We are now ready to begin our search for mixed Nash equilibria.
In this section, we consider the case where $p$ is moderately small. Specifically, we assume $p = n^{-1 + \delta}$,
where $\delta=\delta(n)$ satisfies
\begin{align}
\liminf_{n \to \infty}
\dfrac{ \delta \log n}{\log \log n}
> 8, 
\quad 
\delta 
\le \dfrac13
\bR{1
- \dfrac{4}{\log n}},
\mbox{ and }
\delta \not\in
\bR{ \dfrac14 - \dfrac{3}{2\log n},
\dfrac14 + \dfrac{2}{\log n} }. 
\label{del-Cond}
\end{align}
For a given $p$ in this range, we first carefully identify a specific integer $\ell \ge 1$, 
and execute the generic algorithm with parameter $\ell$ to search for a MNE that is supported on an $\ell\times\ell$ subgame.
We remark that those $\delta$ excluded in the third condition appearing in \eqref{del-Cond} are covered in Section~\ref{sec:gap}, where we give algorithms for the two gap intervals described in Table~\ref{tab:results}.

\subsection{The choice of $\ell$ }
So our first task is to judiciously select $\ell$.
Observe that the constraints in \eqref{del-Cond} imply
\begin{align}
\label{eq:nd}
(\log n)^{8+\ep_0}
\le n^{\delta}
\le
n^{1/3} e^{-4/3}
\end{align}
for some $\eps_0 = \eps_0(n)>0$ 
that is bounded away from zero.
Let
\begin{align}
\calL(\delta)
:= \bC{
\ell \in \mathbb N: \dfrac{1}{2\ell}
\bR{ 1 + \dfrac{4\ell}{\log n} }
\le
\delta
\le \dfrac{1}{1+\ell}
\bR{ 1 - \dfrac{2\ell}{\log n} }
}.
\label{Ldef}
\end{align}
\begin{tcolorbox}
\begin{lm}\label{le:delta}
For all $\delta$ satisfying \eqref{del-Cond},
the set
$
\calL(\delta)$
is non-empty.
\end{lm}
\end{tcolorbox}
\begin{proof}
Define
\begin{align}
L_1 = \dfrac1{2\delta} \bR{ 1 - \dfrac{2}{\delta \log n} }^{-1}
\quad\mbox{ and }\quad
L_2 = \dfrac{ \log n + 2 }{ \delta \log n + 2} - 1.
\label{L12def}
\end{align}
It then follows that
$
\calL(\delta) 
= 
\bC{
\ell \in \mathbb N: 
L_1 
\le \ell \le L_2
}.
$
We remark, for later use, that $L_2 \le 2/\delta$. 
Let's now verify that $\calL(\delta)$ is non-empty. By definition, if $1/4 + 2/\log n
\le \delta 
\le 
\bR{1
- 4/\log n}/3$,
then $2 \in \calL(\delta)$.
Similarly, if 
$
(1 +12/\log n)/6
\le \delta
\le (1 - 6/\log n)/4
$,
then $3 \in \calL(\delta)$.
For $\delta\le1/5$, we have
$L_2 - L_1 > 1$,
and therefore the interval $(L_1, L_2)$ contains at least one integer.
\end{proof}

%
%
%
%
%

For given $\delta$, in the range specified by \eqref{del-Cond},
we may now choose any $\ell \in \calL(\delta)$ and apply the generic algorithm with parameter $\ell$.
Our main result in this section is to bound the expected running time.

\begin{tcolorbox}
\begin{tm} 
\label{tm:cycle1}
Let $p = n^{-1 + \delta}$, where $\delta$ satisfies \eqref{del-Cond},
and let $\ell \in \calL(\delta)$. Then
the expected running time of the generic algorithm with parameter $\ell$ is $\O{ n^2 }$.
\end{tm}
\end{tcolorbox}


To analyze the expected running time of the generic algorithm with parameter $\ell$,
we show that, with high probability,
the random bipartite digraph
has many vertex-disjoint $2\ell$-cycles,
and, moreover, 
the probability that 
all disjoint cycles are unstable is small.
Then we conclude that
there exists a stable $2\ell$-cycle that gets found by the generic algorithm with parameter $\ell$
with sufficiently high probability.

\subsection{The probability that all disjoint cycles are unstable}

First, we bound the probability of
having many disjoint unstable $2\ell$-cycles.

\begin{tcolorbox}
\begin{lm}\label{le:bonc} 
For integer $\ell \ge 1$ and $\beta\ell < 1$,
\begin{align*}
\P{ \mathcal{D}_n \text{ has at least $\beta n$ disjoint unstable $2\ell$-cycles} }
\le
\bR{\frac{e\ell}\beta n^{2\ell} p^{2\ell+2} }^{\beta n}.
\end{align*}
\end{lm}
\end{tcolorbox}
\begin{proof}
Let $m := \lceil \beta n \rceil$.
Fix two vertex sets in different parts $A \subset R$ and $B \subset C$ with $|A| = |B| = \ell$. 
The union bound gives
\begin{align*}
&\P{ \mathcal{D}_n \text{ has at least $\beta n $ disjoint unstable $2\ell$-cycles} } \\
&\le
\binom{n}{\ell m}^2
\bR{ \dfrac{ (\ell m)! }{ \ell!^m m! } }^2 m!
\,
\P{\D \mbox{ has an unstable cycle on $A\cup B$}}^m.
\end{align*}
This follows by noting that $\binom{n}{\ell m}$ is the number of ways to choose $\ell m$ vertices 
from $R$ and $C$
to form a set of $m$ pairwise vertex-disjoint $2\ell$-cycles.
For each part, 
there are $(\ell m)! / (\ell!^m m!) $ ways
to divide $\ell m$ vertices into $m$ groups of $\ell$ vertices.

Then a $2\ell$-cycle is formed by, from each part of
bipartite graph, choosing one of the $m$ groups first,
and then choosing vertices from $\ell$ vertices in a group.
Next, the number of perfect matchings 
and the number of directed Hamiltonian cycles 
(cycles that visit each vertex exactly once)
of a complete bipartite graph $K_{\ell, \ell}$
equal $\ell!$ and $\ell! (\ell-1)!$, respectively.
Note that the cycles are pairwise vertex-disjoint in $\D$ by construction,
and therefore, the occurrences of arcs on the different cycles are independent.
Recalling that a cycle is unstable if there is a common out-neighbour of some pair of two vertices of the cycle,
by the union bound, we have
\begin{align*}
\P{\D \mbox{ has an unstable cycle on $A\cup B$}}
\le 
2 \bR{ \ell ! (\ell-1)! }^m \binom{\ell}{2} np^{2\ell+2},
\end{align*}
where we choose two vertices from the $\ell$ vertices on one part of $\D$ to share a common out-neighbour,
whose size is upper bounded by $n$.

Summarizing, we have
\begin{align*}
\P{ \mathcal{D}_n \text{ has at least $\beta n $ disjoint unstable $2\ell$-cycles} }
&\le
\dfrac{ n^{2\ell m} }{ m! }
\bR{ \ell np^{2\ell+2} }^m 
\le
\bR{ \dfrac{ e \ell }{m} n^{2\ell+1} p^{2\ell+2} }^m,
\end{align*}
which completes the proof in view of the definition of $m$.
\end{proof}

\subsection{The existence of a partial cycle-factor}

Next, we bound the probability 
that $\D$ contains less than
$n/(2\ell)$
vertex-disjoint $2\ell$-cycles,
which we call a {\it partial cycle-factor} of $\D$
(see, e.g., \cite[Section 4.2]{JLR00}
for more general graph factors).
It is useful to define
\begin{equation}\label{eq:deftau}
\tau := \frac{1}{2 \ell}.
\end{equation}

Then
\begin{align}
&\P{ \mathcal{D}_n \text{ has less than $\tau n$ disjoint $2\ell$-cycles} } \notag\\
&=\P{
\mbox{the largest set of disjoint } 2\ell\mbox{-cycles in } \D
\mbox{ is of size}
< \tau n
} \notag\\
&=\P{
\exists n/2 \mbox{ vertices in both } R, C
\mbox{ with no induced } 2\ell\mbox{-cycles in } \D
} \notag\\
&\le
\binom{n}{ \lceil n/2 \rceil }^2
\P{ Y_{\lceil n/2 \rceil} = 0 },
\label{X=0}
\end{align}
where $Y_{\lceil n/2 \rceil}$ 
counts the $2\ell$-cycles in the random bipartite digraph $\mathcal D_{\lceil n/2 \rceil}$ on $\lceil n/2 \rceil \times \lceil n/2 \rceil$ vertices.
Indeed,
if there are at most $\tau n$ disjoint $2\ell$-cycles, 
then, on each part, at most $\ell \tau n = n/2$ vertices are  used on the disjoint $2\ell$-cycles.
Then
there are more than
$ (1 - \tau \ell)n = n/2$
vertices in both $R$ and $C$ that are not covered by any $2\ell$-cycles.
\begin{tcolorbox}
\begin{cl}\label{cl:X=0}
Assume that $\ell = \o{ \sqrt{np} }$.
Then
\begin{align*}
&\P{ \mathcal{D}_n \text{ has less than $\tau n$ disjoint $2\ell$-cycles} } \notag\\
&\qquad\qquad\le 
\dfrac{4^n}{ n} \Exp{
- \dfrac{ 1 }{6\ell \cdot 2^{2\ell}} \bR{ n p }^{2\ell} }
+ 
\dfrac{4^n}{ n} \Exp{ 
- \dfrac{1}{150\ell^8} n^2 p }
\end{align*}
for sufficiently large $n$.
\end{cl}
\end{tcolorbox}
To prove this claim,
we use Janson's inequality \cite[Theorem 2.18]{JLR00} to bound
$\P{ Y_{\lceil n/2 \rceil} = 0 }$ in \eqref{X=0}.
In particular, Janson's inequality implies the following.

\begin{tcolorbox}
\begin{lm}
\label{lm:janson}
Let
$\{ C_i \}_i$ be the set of
all directed $2\ell$-cycles in the complete bipartite digraph 
$K_{\lceil n/2 \rceil, \lceil n/2 \rceil}$
on $\lceil n/2 \rceil \times \lceil n/2 \rceil$ vertices,
and for each $i$, let $X_i$ be the indicator function of event $\{ C_i \subset \mathcal D_{\lceil n/2 \rceil} \}$.
Then
\begin{align*}
\P{ Y_{\lceil n/2 \rceil} = 0 }
=
\P{ \sum_i X_i = 0 }
\le 
\Exp{ 
- \dfrac{\mu^{2}}{\mu + 2\Delta} },
\end{align*}
where
$
\mu
:= \sum_{ i } \E{ X_i }
$, and 
$\Delta 
:=
\sum_{i,j \colon i \sim j } \E{ X_i X_j }$,
where we write $i \sim j$
if $i \ne j$ and two directed cycles $C_i$ and $C_j$ share 
at least one directed edge 
in the edge set $E( K_{\lceil n/2 \rceil, \lceil n/2 \rceil} )$.
\end{lm}
\end{tcolorbox}

The next lemma is useful for proving Claim \ref{cl:X=0}.
\begin{tcolorbox}
\begin{lm}\label{hyper}
For any $p>0$ and integers $N, \ell$ such that $N > \ell > 0$, we have
\begin{align*}
\sum_{s \in [\ell]}
\binom{\ell}{s}
\binom{N-\ell}{\ell-s} p^{-s}
\le 
\binom{N}{\ell} \bR{ 1 + \dfrac{\ell}{Np} }^\ell
- \binom{N-\ell}{\ell}.
\end{align*}
\end{lm}
\end{tcolorbox}
\begin{proof}
Let $Z$ be a hypergeometric random variable with parameters $N, \ell, \ell$.
Then
\begin{align}\label{eq:expl}
\sum_{s \in [\ell]}
\binom{\ell}{s}
\binom{N-\ell}{\ell-s} p^{-s}
= 
\binom{N}{\ell}
\sum_{s \in [\ell]}
\dfrac{ \binom{\ell}{s}
\binom{N-\ell}{\ell-s} }
{ \binom{N}{\ell} } p^{-s} 
= \binom{N}{\ell} \E{ p^{-Z} \I{ Z \ge 1 }}.
\end{align}
Since 
the moment-generating function of a hypergeometric distribution 
is not greater than that of the binomial distribution with the same mean (see \cite[Theorem 2.10]{JLR00}),
we have
\begin{align*}
\E{ p^{-Z} }
\le 
\bR{ 1 + \dfrac{\ell}{N} \bR{\dfrac1p-1} }^\ell
\le \bR{ 1 + \dfrac{\ell}{Np} }^\ell.
\end{align*}

By noting
$$
\E{ p^{-Z} } = \E{ p^{-Z} \I{ Z \ge 1 }}
+ \dfrac{
\binom{N-\ell}{\ell} }
{ \binom{N}{\ell} }
$$
and using \eqref{eq:expl}, we obtain the result.
\end{proof}

Now we are ready to prove Claim \ref{cl:X=0}. \\
\begin{proof}[Proof of Claim \ref{cl:X=0}]
We use $[n]_{j}$ to denote the falling factorial $n (n-1) \cdots (n-j+1)$. Recall that $Y_{\lceil n/2 \rceil}$ counts the $2\ell$-cycles in the random bipartite digraph on $\lceil n/2 \rceil \times \lceil n/2 \rceil$ vertices
and $\mu = \E{ Y_{\lceil n/2 \rceil} }$.
Then
\begin{align*}
\mu
= \dfrac{ \bR{\ell!}^{2}}{2\ell} \binom{\lceil n/2\rceil}{\ell}^{2} p^{2\ell}
= \dfrac{ [\lceil n/2\rceil]_{\ell}^{2} }{2\ell} p^{2\ell},
\end{align*}
where the factor $\ell!$ is used to choose the ordering of the vertices on one part,
and $2\ell$ is due to the fact that every vertex on a cycle can be the starting vertex for counting. 
For $\ell = \o{ \sqrt n }$, 
we have
\begin{align}
\mu
\ge 
\dfrac{ 1 }{2\ell} ( p(n/2 - \ell) )^{2\ell}
&= \dfrac{ 1 }{2\ell} 
\bR{ 1 - \dfrac{\ell}{n/2} }^{2\ell}
( np/2 )^{2\ell} \notag\\
&= (1+o(1)) \dfrac{ 1 }{2\ell} 
( np/2 )^{2\ell} 
\ge \dfrac{ 1 }{3\ell} 
( np/2 )^{2\ell},
\label{mu}
\end{align}
for all sufficiently large $n$.

Fix $\i$ (resp. $\j$) to be a set of $\ell$ vertices in $R$ (resp. $C$).
Define $X_{\i, \j}$ to be the indicator function of the event that all the vertices in $\{\i, \j\}$ are part of the same $2\ell$ cycle. 
Denote by $B = \binom{[ \lceil n/2 \rceil ]}{\ell}$ the set of all size-$\ell$ subsets of $[ \lceil n/2 \rceil ]$. For $(\i, \j), (\i', \j') \in B^2$, we use $(\i, \j) \sim (\i', \j')$ to denote $(\i, \j) \neq (\i', \j')$ and $|\i \cap \i'|\vee |\j \cap \j'| \ge 1$,
where $a \vee b$ denotes the maximum of $a$ and $b$.
Then we have
\begin{align*}
\Delta 
= 
\sum_{
(\i, \j), (\i', \j') \in B^2:
(\i, \j) \sim (\i', \j')} 
\E{ X_{\i, \j} X_{\i', \j'}}.
\end{align*}
Using again the fact that the number of directed Hamiltonian cycles for the
complete bipartite graph $K_{\ell, \ell}$ is $\ell! (\ell-1)!$,
we have
\begin{align*}
\Delta 
&\le
\sum_{s \in [\ell]}
\sum_{t \in [\ell]} 
\binom{\lceil n/2 \rceil}{\ell}^{2}
\binom{\ell}{s}
\binom{\lceil n/2 \rceil-\ell}{\ell-s}
\binom{\ell}{t}
\binom{\lceil n/2 \rceil-\ell}{\ell-t}
\bR{ \ell! (\ell-1)! }^{2} 
p^{4\ell- (s+t-1)} \\
&=
\binom{\lceil n/2 \rceil}{\ell}^{2}
\bR{ \ell! (\ell-1)! }^{2} p^{4\ell+1}
\bR{ \sum_{s \in [\ell]}
\binom{\ell}{s}
\binom{\lceil n/2 \rceil-\ell}{\ell-s} p^{-s} }^2,
\end{align*}
where the first inequality
follows by noting that
for two directed cycles sharing $s+t$ vertices,
the number of shared arcs is at most $s+t-1$.

Using Lemma \ref{hyper} with $N= \lceil n/2 \rceil$, and noting that
\begin{align*}
\binom{x-\ell}{\ell}
= \dfrac{[x-\ell]_\ell}{\ell!}
\ge \dfrac{\bR{x-2\ell}^\ell}{\ell!}
= \dfrac{x^\ell}{\ell!} \bR{1-\dfrac{2\ell}{x}}^\ell
\end{align*}
for all $x > \ell \ge 1$,
we have
\begin{align*}
\Delta 
&\le 
\binom{\lceil n/2 \rceil}{\ell}^{2}
\bR{ \ell! (\ell-1)! }^{2} p^{4\ell+1}
\bR{ \binom{\lceil n/2 \rceil}{\ell} \bR{ 1 + \dfrac{\ell}{\lceil n/2 \rceil p} }^\ell
- \binom{\lceil n/2 \rceil-\ell}{\ell} }^2 \notag \\
&\le 
\lceil n/2 \rceil^{4\ell}
p^{4\ell+1}
\bR{ \bR{ 1 + \dfrac{\ell}{np/2} }^\ell
- \bR{ 1- \dfrac{2\ell}{n/2}}^{\ell} }^2 \notag.
\end{align*}

Using ${\ell \choose i} \le \ell^i$, combined with the binomial theorem, we have 
\begin{align}
\bR{ 1 + \dfrac{\ell}{np/2} }^\ell
- \bR{ 1- \dfrac{2\ell}{n/2}}^{\ell}
&= 
\sum_{i = 0}^\ell {\ell \choose i}\bR{ \dfrac{\ell}{np/2} }^i
- \sum_{i = 0}^\ell {\ell \choose i} \bR{ \dfrac{-2\ell}{n/2} }^i \notag\\
&\le 
\sum_{i = 1}^\ell \bR{ \dfrac{\ell^2}{np/2} }^i
+ \sum_{i = 1}^\ell \bR{ \dfrac{2\ell^2}{n/2} }^i. \notag
\end{align}
Hence, 
by noting that $\lceil n/2 \rceil^{4\ell} / ( n/2 )^{4\ell}
\le \bR{1+ 2/n}^{4\ell}
\le e^{8\ell/n}
$, we have 
\begin{align}
\Delta 
&\le 
\bR{ n/2}^{4\ell}
e^{8\ell/n}
p^{4\ell+1}
\bR{ \sum_{i = 1}^\ell \bR{ \dfrac{\ell^2}{np/2} }^i
+ \sum_{i = 1}^\ell \bR{ \dfrac{2\ell^2}{n/2} }^i }^2 \notag\\
&\le 
e^{8\ell/n}
\bR{ n/2}^{4\ell}
p^{4\ell+1}
\bR{ \dfrac{2\ell^3}{np/2} }^2
= 4\ell^6 
e^{8\ell/n}
\bR{ n/2}^{4\ell-2}
p^{4\ell-1},
\label{delta}
\end{align}
where the second inequality holds for all $n$ large enough, and by noting $\ell = \o{ \sqrt {np} }$.

Let
$\newconstant{co:co5}$
\begin{align*}
\useconstant{co:co5} := 12\ell^7
e^{8\ell/n}
(n/2 )^{2\ell-2} p^{2\ell-1}.
\end{align*}
Using \eqref{mu} and \eqref{delta}, we have
$
\Delta/\mu
\le \useconstant{co:co5}.
$
Next, we distinguish two cases.\\
{\bf Case 1}: If $\useconstant{co:co5}\le1/2$ then, by Lemma~\ref{lm:janson} and \eqref{mu},
\begin{align*}
\P{ Y_{\lceil n/2 \rceil} = 0 }
\le \Exp{ 
- \dfrac{\mu}{1 + 2\Delta / \mu} }
\le \Exp{ 
- \dfrac{\mu}{1 + 2\useconstant{co:co5}} }
\le \Exp{- \dfrac{\mu}{2} }
\le \Exp{
- \dfrac{ 1 }{6\ell} 
( np/2 )^{2\ell} }.
\end{align*}
{\bf Case 2}: If $\useconstant{co:co5} > 1/2$ then, by Lemma~\ref{lm:janson} and \eqref{mu},
\begin{align*}
\P{ Y_{\lceil n/2 \rceil} = 0 }
\le \Exp{ 
- \dfrac{\mu}{1 + 2\Delta / \mu} }
&\le \Exp{ 
- \dfrac{\mu}{4\useconstant{co:co5}} }\\
&\le \Exp{ 
- \dfrac{1}{36 \ell^8}
( n/2 )^2 p e^{-8\ell/n} } 
\le \Exp{ 
- \dfrac{1}{150\ell^{8}}
n^{2} p },
\end{align*}
by recalling that $\ell=o(n)$ as $n\to\infty$.
Now, for all integers $n\ge1$, note that
\begin{align*}
\binom{n}{ \lceil n/2 \rceil }
\le
\dfrac{2^{n+1}}{\sqrt{ \pi (2n+1) }}
\le 
\dfrac{2^n}{\sqrt n}.
\end{align*}
Combining the above, the proof is then complete in view of
\eqref{X=0}.
\end{proof}

\subsection{An efficient algorithm for Regime 1}

%
%

Now we are ready to prove Theorem \ref{tm:cycle1}, that our generic algorithm gives an efficient algorithm for Regime 1.


\begin{proof}[Proof of Theorem \ref{tm:cycle1}]
First, note that the algorithm is guaranteed to find a Nash equilibrium.
Recall the Lemke-Howson constant $\clh \le 2.598$ in \eqref{eq:clh} and 
$\tau = (2 \ell)^{-1}$ defined by \eqref{eq:deftau}.
Bounding the probability of not having a pure Nash equilibrium by \eqref{pne},
we have that the expected running time of the generic algorithm with parameter $\ell$ is at most
\begin{align}
& O(n^2) 
+ \O{\binom{n}{\ell}^2
\poly(n) } \P{ \D \mbox{ has no PNE} }
+ \O{ \clh^n } \P{ \mathcal{D}_n \text{ has no 
stable $2\ell$-cycle} } \notag\\
&= O(n^2) 
+ \O{ T_0(\ell) } + \O{ T_1(\ell) } + \O{ T_2(\ell) },
\label{time1}
\end{align}
where, using Lemma \ref{Pfail} and Claim \ref{cl:X=0},
$$
T_0 ( \ell ) := e^{-n^2 p^2} \binom{n}{\ell}^2 \poly(n), \qquad 
T_1( \ell ) := \clh^n \P{\mathcal{A}_\ell^c}, \qquad 
T_2( \ell ) := \clh^n \P{ \mathcal{C}_{\ell}},
$$
and $\mathcal{A}^c_\ell$ is the event that $\mathcal{D}_n$ does contain at least $\tau n$ disjoint $2\ell$-cycles,
and $\mathcal{C}_{\ell}$ is the event that $\mathcal{D}_n$ contains at least $\tau n $ disjoint unstable $2\ell$ cycles.

We emphasize the need to search for PNE in the first step of the algorithm. This is because $\ell$ could grow in $n$, 
in view of the definition $\calL(\delta)$ in \eqref{Ldef},
for example in the case when $\delta = (8+\ep_0)\log\log n / \log n$.

$\bullet$ We first analyze $T_0(\ell)$. 
Note that $np = n^{\delta}$ and $\ell \le L_2 \le 2/\delta$ (see \eqref{L12def}). 
Then
\begin{align*}
T_0(\ell)
= \O{ \Exp{ - n^{2\delta} + \O{ \ell \log n } } }.
\end{align*}
In view of \eqref{eq:nd},
we have $n^{2\delta}\ge(\log n)^{16+2\ep_0}$,
and therefore,
\begin{align*}
\ell \log n
\le 
\dfrac{2}{\delta}\log n
\le 
\dfrac{ (\log n)^2 }{2 \log\log n}
=
o( n^{2\delta}).
\end{align*}

$\bullet$ Next, we focus on $T_2(\ell)$.
By Lemma \ref{le:bonc}, 
we have, by noting $\tau \ell =1/2$, that
\begin{align*}
T_2(\ell) 
= \O{ \clh^n \bR{ \dfrac{e \ell}{\tau } n^{2(\delta(1+\ell)-1)} }^{\tau n} }
= \O{ \clh^n \bR{ 2 e \ell^2 n^{2(\delta(1+\ell)-1)} }^{n/2\ell}}.
\end{align*}
We use $
\delta
\le \frac{1}{1+\ell} \bR{ 1 - \frac{2\ell}{\log n} }
$ to obtain that
\begin{align*}
\bR{ 2 e \ell^2 n^{2(\delta(1+\ell)-1)} }^{1/2\ell}
&= 
\Exp{ - \dfrac{1}{\ell} 
\bR{ (1 - \delta(1+\ell)) \log n 
- \dfrac12 (1+\log2) - \log \ell } } \\
&\le \Exp{ - \dfrac{1}{\ell} 
\bR{ 2\ell
- \dfrac12 (1+\log2) - \log\ell } }
\le \Exp{-1} < \frac 1{\clh},
\end{align*}
when $\ell \ge 2$, which implies $T_2(\ell)=o(1)$.

$\bullet$ Now we bound $T_1( \ell )$ using Claim \ref{cl:X=0}.
First, note that
\begin{align*}
\ell 
= \O{ \dfrac{\log n}{\log\log n} }
= \o{ n^{\delta/2} }
= \o{ \sqrt {np} },
\end{align*}
in view of the bound \eqref{eq:nd} on $n^\delta = np$.
Then, using Claim \ref{cl:X=0},
\begin{align*}
T_1( \ell ) 
&\le 
\dfrac{1}{ n} \bR{4\clh}^n
\Exp{ 
- \dfrac{ 1 }{6 \ell } 
\bR{ n p/2 }^{2 \ell } } 
+ \dfrac{1}{ n} \bR{4\clh}^n
\Exp{ 
- \dfrac{1}{150 \ell^8} n^2 p }
\\
&=
\Exp{ 
- n f_1( \ell )}
+ \Exp{ 
- n f_2( \ell )},
\end{align*}
where we can rewrite, using $p = n^{\delta -1}$, 
\begin{align*}
f_1( \ell )
&= \dfrac{ 1 }{6 \ell } 
\bR{ p/2 }^{2 \ell }
n^{2 \ell -1} - \log4\clh + \dfrac{\log n}{n} =
\dfrac{ 1 }{6 \ell } 
2^{-2 \ell }
n^{ 2\delta \ell - 1} - \log4\clh + \dfrac{\log n}{n},
\end{align*}
and
\begin{align*}
f_2( \ell )
&= \dfrac{1}{150 \ell^8} n p - \log4\clh + \dfrac{\log n}{n}.
\end{align*}
We first prove that $f_2( \ell ) >0$, for all $n$ large enough. Recall 
$\delta \ge 8 \log\log n / \log n.$
As $\ell \le L_2\le 2/\delta$, where $L_2$ is defined by \eqref{L12def}, we have
\begin{align*}
\dfrac{1}{150 \ell^8} n p
&\ge
\dfrac{ 1 }{150 \cdot 2^8 } \delta^8 n^{\delta} 
=
\Omega( (\log\log n)^{8+\ep_0} )
> \log4\clh.
\end{align*}

Note that
\begin{align*}
\dfrac{ 1 }{6 \ell } 
2^{-2 \ell }
n^{ 2\delta \ell - 1}
= \Exp{ 
- 2 \ell \log 2 - \log(6 \ell )
+ (2\delta \ell - 1) \log n
},
\end{align*}
and we show that $f_1( \ell )>0$ by
proving
\begin{align}
2\delta \ell - 1 
> \dfrac{ 2 \ell \log 2 + \log(6 \ell ) }{ \log n }
+ \dfrac{1}{\log n} \Log{ \log4\clh - \dfrac{\log n}{n} }.
\label{Lmin}
\end{align}
This follows by noting, from \eqref{Ldef}, that
$
2\delta \ell > 1 + 4 \ell /\log n,
$
which implies \eqref{Lmin} for all $\ell\ge2$.
The proof is complete 
in view of the bound on the running time \eqref{time1}.
\end{proof}

\section{Regime 2: disjoint stable 4-cycles with localized pair}
\label{sec:reg2}
In this section, 
we study the case
\begin{align}\label{local-p}
2^9 n^{-2/3} \le p \le  e^{-52} 2^{-8}  n^{-1/2}.
\end{align}
In this range, we use the generic algorithm with parameter $\ell=2$.
To show that this algorithm works, we need to prove that a
stable $4$-cycle exists with sufficiently high probability.
To this end,
we search for disjoint $4$-cycles with an additional local constraint
to guarantee that the probability of having disjoint unstable cycles is small.
In order to prove the existence of an algorithm that finds a Nash equilibrium in expected polynomial time, 
we show that $\D$ contains a stable $4$-cycle with sufficiently high probability.
This leads to our main result in this section.
\begin{tcolorbox}
\begin{tm}\label{thm-local}
For $2^9 n^{-2/3} \le p \le  e^{-52} 2^{-8}  n^{-1/2}$, the expected running time of the generic algorithm with parameter $\ell=2$
is $\O{n^2}$.
\end{tm}
\end{tcolorbox}

First, recall that the two parts of the bipartite graph $\D$ are
$R = (r_1, \ldots, r_n)$ 
and $C=(c_1, \ldots, c_n)$.
For any two distinct vertices $r_i, r_j \in R$,
we say that the pair $(r_i , r_j)$ is {\it bad} if there are at least three vertices in $C$
that are the common out-neighbours of $r_i $ and $r_j$. 
Note that if a pair is bad, it cannot be the vertices of a stable cycle
by Definition \ref{def:unst}.
It is useful to recall the following Chernoff bound of the binomial tail, see, e.g., \cite[Corollary 2.4]{JLR00}.
\begin{lm}
\label{tail}
For any $t>np$, we have
\begin{align*}
\mathbf P( {\rm Bin}(n, p) \ge t )
\le e^{-np}
\bR{\frac{enp}{t}}^t.
\end{align*}
\end{lm}


Using Lemma \ref{tail}, and noting $np^2 \le 1$, we have
\begin{align}
\P{ (r_i , r_j) \text{ is bad} }
\le \P{ {\rm Bin}(n, p^2) \ge 3 } 
\le e^{-np^2} \bR{ \dfrac{enp^2}{3} }^{3}.
\label{badprob}
\end{align}
The locality constraint we impose is on the distance between the indices of a pair of distinct vertices in $R$,
more specifically, we focus on pairs $(r_i , r_j)$ such that
\newconstant{para}
\begin{align}
\abs{ j - i } \le 
2^{20} \useconstant{para} 
=: d(n, p),
\quad\mbox{ where }
\useconstant{para}
= \useconstant{para}(n, p)
:=  1/(n^2p^4).
\label{k2def}
\end{align}
We call these {\it localized pairs}.
The probability of having 
a set of disjoint $d(n, p)$-localized pairs 
that are all unstable, 
where the number of cycles is at least $t =n/100$, is
\begin{align}
&\P{ 
\D \text{ has at least $n/100$ disjoint unstable $d(n, p)$-localized pairs}  } \notag\\
&\le
\binom{n}{t} 
\left(d(n, p) \right)^t
\P{ (r_i , r_j) \text{ is bad} }^t \notag\\
&\le \bR{ \frac{en}{t} }^t
\left( \frac{2^{20}}{n^2 p^4} \right)^t
\bR{ \dfrac{enp^2}{3} }^{3t} 
= \left( \frac{ e^4 2^{20} n^2 p^2 }{ 27t } \right)^t \le e^{-n}.
\label{PallBad}
\end{align}
Here the first inequality follows because we may pick the first vertex of $R$ in the cycle in 
$\binom{n}{t}$ ways and then the second vertex of $R$ in the cycle must be within distance $d(n, p)$ from it.
We used \eqref{badprob} in the second inequality,
and in view of the range of $p$ in \eqref{local-p}.

\subsection{The pairing procedure}

We will show that, with high probability, $\D$ has either at least $n/100$ disjoint  $4$-cycles or a $2$-cycle 
(corresponding to a PNE). 
This will be proved using a \emph{pairing procedure} described next.

\begin{tcolorbox}[float,floatplacement=!tbh, title={Pairing Procedure: searching for disjoint $4$-cycles.}]
\begin{algorithm}[H]
\textbf{Input:} 
Bipartite digraph $\D$
with vertex bipartition
$R = ({r}_1, \ldots, {r}_n)$ 
and $C=({c}_1, \ldots, {c}_n)$,
and the set of arcs $E(\D)$.
\\
\textbf{Output:} 
A set $ \mathcal Z$, containing either disjoint cycles $\{ ( {r}_i , {c}_j, {r}_k, {c}_\ell ) \}$ or a PNE.

The initialization is by:
(i1) letting $\mathcal{Z} = \emptyset$;
(i2) letting $E = E(\D)$ contain the set of available arcs in $\D$ that are not explored;
(i3) letting
$\mathfrak{R} = R$ and
$\mathfrak{C} = C$ 
contain the vertices that are candidates of the pairing process in $R$ and $C$, respectively.
\\
\While{ $ | \mathfrak{R} | \ge 2 $ 
and $\mathcal{Z}$ has no PNE } {
Let $r_i$ be the first vertex in $\mathfrak{R}$, i.e., the vertex with the smallest index.\\
Expose the first $m(n, p)$ arcs in $E$ 
outgoing from and incoming to ${r}_i$. 
Denote by $\widetilde{N}^+_E( {r}_i) $ and 
$\widetilde{N}^-_E( {r}_i) $ these sets of arcs. \\
\For{$r_j \in \mathfrak{R}$ with $0 < j-i \le d(n,p)$}
{
{
\For{
${c}_p \in \widetilde{N}^+_E( {r}_i)$ and ${c}_q \in \widetilde{N}^-_E( {r}_i)$ }
 {
 Check if any of the vertices explored form a PNE? If yes, stop, and return $\mathcal{Z}$; if not, continue.\\
{
\If{ $ {r}_j \in \widetilde{N}^+_E({c}_p)$ and $ {r}_j \in \widetilde{N}^-_E({c}_q) $ }{

Add $( {r}_i , {c}_p, {r}_j, {c}_q )$ to set $\mathcal{Z}$, 
remove ${r}_i, {r}_{i+1}, {r}_j$ from $\mathfrak{R}$,
and remove ${c}_p, {c}_q$ from $\mathfrak{C}$;

Remove arcs
$\cup_{v \in \widetilde{N}^+_E({c}_p)} ( {c}_p, v ) $
and $\cup_{v \in \widetilde{N}^-_E({c}_q)} ( v, {c}_q ) $ from $E$;

Remove
$\cup_{u \in \widetilde{N}^+_E({r}_i )} ( {r}_i , u )$
and 
$\cup_{u \in \widetilde{N}^-_E( r_i )} ( u, {r}_i ) $
from $E$; \\
Restart the while loop from line 4.
}
}
}
}

} 
Remove
$\cup_{u \in \widetilde{N}^+_E({r}_i )} ( {r}_i , u )$
and $\cup_{u \in \widetilde{N}^-_E( r_i )} ( u, {r}_i ) $ from $E$;

{
$\mathfrak{R} \gets \mathfrak{R} \setminus \{{r}_i \}$ and restart the loop from line 4. }
}
\Return{$\mathcal{Z}$}
\end{algorithm}
\end{tcolorbox}

\begin{tcolorbox}
\begin{cl}
\label{pairing}
Let $\mathcal{Z}$ be the output of the pairing procedure, and define the event 
$$
J := \{\mathcal{Z}\mbox{ contains at least $n/100$ disjoint  $4$-cycles or a PNE}\,\}.
$$
Then for $p$ satisfying \eqref{local-p}, we have 
$
\P{J^c} 
\le
e^{ - 2n }
$,
where $J^c$ denotes the complement of event $J$.
\end{cl}
\end{tcolorbox}

\begin{proof}
To find disjoint $4$-cycles,
we explore the vertices in $R$ sequentially.
For any two distinct vertices $r_i, r_j \in R$, 
we say that they are {\it paired} if they are in a directed $4$-cycle in $\D$. Let's now explain how the pairing procedure works.\\
$1)$ At each step, the pairing procedure focuses on a given vertex in $R$, say $r_i$. 
Either this vertex is successful, i.e., can be paired with another vertex $r_j$ that is within distance $d(n,p)$ from $r_i$, 
and satisfies certain rules described below, or it is unsuccessful, 
in which case the algorithm removes $r_i$, deletes all the arcs that have been revealed, and moves to $r_{i+1}$. 
If $r_i$ is successful and paired with $r_j$, 
the triplet $(r_i, r_{i+1}, r_j)$ is deleted from $R$, 
together with all the arcs that have been exposed so far, and the two vertices in $C$ that are part of the same 4-cycle of $(r_i, r_{j})$.\\ 
$2)$ We describe next what it means for $r_i $ to be a success. 
The task is to pair $r_i$ with some other vertex $r_j $ with $0 < j - i \le d(n, p)$.
Some of these vertices have been removed from previous iterations of the algorithm, 
but 
the fact that vertex $r_{i+1}$ is deleted when $r_i$ is successful
guarantees that there are at least $d(n, p)/2$ vertices that are still available
for the pairing attempts of vertices $\{ r_k : 0 < k - i \le d(n, p) \}$.

We impose further constraints on the number of arcs that can be exposed in the pairing attempts for $r_i$. 
Denote by $E_i$ the set of directed edges in $\D$ that have not been deleted by the time the algorithm reaches $r_i$. 
Note that for any $i \le j$, we have $E_j \subseteq E_i$.
For each vertex $v \in R \cup C$, denote by $N^+_{E_i}(v) $ and $N^-_{E_i}(v)$ 
the set of out-neighbors and in-neighbors of $v$ in the subgraph induced by the edge set $E_i$, respectively. 
Order the elements of $N^+_{E_i}( r_i) $ and $N^-_{E_i}( r_i)$ using the ordering in $C$, 
and truncate them as follows. Reveal (i.e., expose) only the first 
\begin{align*}
m(n, p) := \left\lceil 4 np \right\rceil
\end{align*}
arcs in $E_i$.
Denote by $\widetilde{N}^+_{E_i}( r_i) $ and 
$\widetilde{N}^-_{E_i}( r_i)$ the truncated neighbors described above.\\
$3)$ For each vertex $r_j$, with $0< j-i \le d(n, p)$, which was not deleted in a previous step of the algorithm, repeat the following. 
For each vertex $c_u \in \widetilde{N}^+_{E_i}( r_i) $ 
and $c_v \in \widetilde{N}^-_{E_i}( r_i) $, 
we check if $c_u \in \widetilde{N}^-_{E_i}( r_j ) $ 
and $c_v \in \widetilde{N}^+_{E_i}( r_i ) $, 
so that they form a 4-cycle. 
If such a cycle $( r_i, c_u, r_j, c_v)$ exists, we declare $r_i$ successful and paired to $r_j$ and the algorithm moves to $r_{i+2}$, after deleting all the vertices in the cycle and the arcs that have been exposed.
If $r_i$ cannot be paired to any of the $r_j$, 
the algorithm moves to $r_{i+1}$, after the deletion of vertices and arcs exposed. 

An illustration of the concepts used in the procedure is shown in Figure~\ref{fig:pairing}.

Next, we show that the algorithm described above works with very high probability. Recall that $a\wedge b := \min(a, b)$.
Using the constraint 
that
the in(out)-neighbors size is at most $m(n, p)$ due to truncation,
and in view of the condition on $p$ in \eqref{local-p},
we have that at most 
\begin{align*}
m(n, p) \cdot d(n, p)
\le 2^{3} np \cdot 2^{20} \useconstant{para}
= \frac{2^{23}}{np^3}
\le n/16
\end{align*}
out-neighbors or in-neighbors of $r_i $ have been exposed by the time the algorithm reaches $r_i$. 
Notice that the vertex deletions in $C$ happen only when 
cycles are detected, and therefore, at most 
$2 \ceil{n/3}$ vertices can be removed from $C$,
as the number of $4$-cycles 
is at most $n/3$,
by noting that in the pairing procedure,
if $r_i$ is paired with $r_j$, 
the triplet $(r_i, r_{i+1}, r_j)$ is deleted from $R$.
Hence, during the pairing process,
for any $r_i$, we have at least
\begin{align}
\label{n/4left}
n - 2  \ceil{\dfrac{n}{3}} - m(n, p) d(n, p) \ge n/4,
\end{align}
vertices left in $C$ for the pairing.
Hence, for each $i$, the 
size of out-neighbors $| \widetilde{N}^+_{E_i}( r_i ) |$
and in-neighbors $| \widetilde{N}^-_{E_i}( r_i ) |$ are 
larger than certain random variables 
$Y^+_i$ and $Y^-_i$, such that 
\begin{align}\label{Ydef}
Y_i^{\bullet} \sim {\rm Bin}\bR{ \left\lfloor \dfrac{1}{4} n\right\rfloor, p } \wedge m(n, p),
\end{align}
where $\bullet \in \{+, -\}$, 
and noting that both $\widetilde{N}^+_{E_i}( r_i )$
and $\widetilde{N}^-_{E_i}( r_i )$ are truncated at $m(n, p)$.

If the size of the common out-neighbors and in-neighbors  
$\abs{ \widetilde{N}^+_{E_i}(r_i ) \cap \widetilde{N}^-_{E_i}(r_i ) } \ge 1$ for some vertex $r_i$, 
then the game has a PNE and the algorithm detects it and returns it as an output. 
In that case, the algorithm stops and $r_i $ is declared to be a success. 
Now we assume $\abs{ \widetilde{N}^+_{E_i}(r_i ) \cap \widetilde{N}^-_{E_i}(r_i ) } =0$.

\begin{figure}[t]
\centering
\begin{tikzpicture}[thick, main/.style = {draw, circle}] 

\foreach \n in {1,2,3, 4} {   
	\node[main] (0\n) at (0,-1.5*\n) { $r_{\n} $ };
}

\foreach \n in {1,2,...,4} {   
	\node[draw, rectangle] (1\n) at (4,-1.5*\n) { $c_{\n} $ };
}

\node[circle,draw=blue!80, line width=1mm] (02) at (0,-1.5*2) { $r_2 $};
\node[circle,draw=blue!80, line width=1mm] (04) at (0,-1.5*4) { $r_4 $};

\node[circle,draw=blue!80, line width=1mm, dashed] (03) at (0,-1.5*3) { $r_3 $};

\draw[onethirdarrow={>}, blue, ultra thick] (02) -- (11); 
\draw[onethirdarrow={>}, blue, ultra thick] (14) -- (02); 
\draw[onethirdarrow={>}, blue, ultra thick] (11) -- (04); 
\draw[onethirdarrow={>}, blue, ultra thick] (04) -- (14);

\draw[onethirdarrow={>}, ultra thick] (14) -- (01); 
\draw[onethirdarrow={>}, ultra thick, dashed] (01) -- (11); 
\end{tikzpicture}
\caption{Cycle $(r_2, c_1, r_4, c_4)$ is selected in the set $\mathcal Z$ of disjoint cycles.
Vertex $r_3$ is removed as $r_2$ is paired.
Vertex $r_1$
is failed to pair,
as arcs $(r_1, c_1)$ and $(c_4, r_1)$ were explored 
during the attempt to pair $r_1$ with $r_4$, 
but arc $(r_1, c_1)$ is missing.
We also have $c_4 \in
\widetilde{N}^-_{E_1}( r_1 ) \cap \widetilde{N}^-_{E_1}( r_2 )$.
}\label{fig:pairing}
\end{figure}
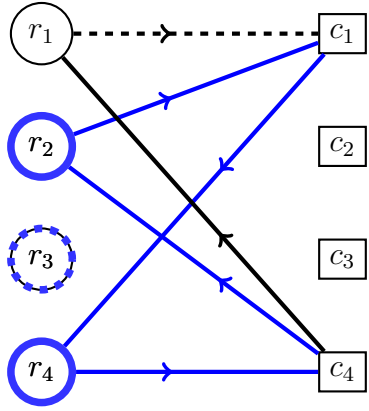

Denote by $\tau_i$ the time the algorithm considers vertex $r_i \in R$ 
and tries to pair it 
with some potential $r_j \in R$,
where  $j$ ranges $0 < j - i \le d(n, p)$. 
We set $\tau_i = \infty$ if either $r_i$  was successfully paired with an $r_s$, with $0< i-s < d(n,p)$, or if $r_{i-1}$ was paired
(so $r_i$ is removed).
There are dependencies among the events that $r_i$ is paired and previous pairing attempts.
We show that with high probability,
vertex $r_i$ is paired.
Suppose that $\tau_i<\infty$. 
Let $U_i$ be any fixed outcome of the algorithm by the time  $\tau_i$. More precisely, $U_i$ is the collection of paired vertices, and arcs and vertices that were deleted by the algorithm by $\tau_i$. As mentioned above, we assume that $r_i$ is not among these vertices that have been paired or removed, otherwise, we would have $\tau_i =\infty$. We will obtain 
the following uniform bound over all $U_i$.




\begin{cl}
\label{ri}
\begin{align}
\P{r_i \text{ is paired}\;|\; U_i} \ge 0.99999.
\label{eq:pivin}
\end{align}
\end{cl}

Using  \eqref{eq:pivin}, together with the fact that there are at least $\lfloor n/3 \rfloor$ pairing trails, 
allows us to conclude that the algorithm either returns a pure Nash equilibrium
or a set of disjoint $4$-cycles, whose size is stochastically larger than ${\rm Bin}(\lfloor n/3 \rfloor, 0.99999)$. 
Hence, the set $\mathcal{Z}$ either has a PNE 
or has size
stochastically larger than a binomial with parameters $\floor{n/3}$ and $0.99999$.
Using the Chernoff bound, 
we have
\begin{align*}
\P{J}
= \P{ \abs{ \mathcal{Z} } \ge n/100 \mbox{ or } \mathcal{Z} \mbox{ has PNEs} }
\ge 1- 
\Exp{ - \lfloor n/3 \rfloor D( 0.03 \| 0.99999 )  },
\end{align*}
where $D$ denotes the relative entropy defined by
\begin{align*}
D\left( q \,\|\, p \right)
:= q \log \dfrac {q}{p} + (1-q) \log \dfrac {1-q}{1-p}.
\end{align*}
Then noting
$$
D( 0.03 \,\|\, 0.99999 )
= 0.03 \log \dfrac {0.03}{0.99999} +
0.97 \log \dfrac {0.97}{0.00001} > 9
$$
completes the proof of Claim~\ref{pairing}.
\end{proof}

It remains to prove Claim \ref{ri}.

\begin{proof}[Proof of Claim \ref{ri}]
Fix a vertex $r_i$. Recall that we use  $\tau_i$ to denote the time the procedure reaches $r_i$ and tries to pair it with potential partners $r_j$ with $0<j-i\le d(n,p)$. 
Below, we focus on $r_j$ with $0< j-i \le d(n, p)$, which is not deleted by the pairing procedure, by time $\tau_i$.

Recall that we have at least  $d(n, p)/2$ such vertices $r_j$ 
to pair with $r_i$,
because of Step 18 of the pairing procedure, that is,
we skip the next vertex after a vertex that is successfully paired.
For each such  $r_j$,
in order for $r_i$ and $r_j$ to be paired, it suffices to have both 
$\abs{ \widetilde{N}^+_{E_i}(r_i )
\cap \widetilde{N}^-_{E_i}(r_j) }$ 
and $\abs{ \widetilde{N}^-_{E_i}(r_i )
\cap \widetilde{N}^+_{E_i}(r_j) }$ to be positive. 
Recalling the definition of $Y^+_i$ and $Y^-_i$
in \eqref{Ydef},
conditioning on event $
\bC{ Y^+_i\wedge Y^-_i \ge n p /12 }$,
we have that both $\abs{ \widetilde{N}^+_{E_i}(r_i )
\cap \widetilde{N}^-_{E_i}(r_j) }$ 
and $\abs{ \widetilde{N}^-_{E_i}(r_i )
\cap \widetilde{N}^+_{E_i}(r_j) }$ 
are stochastically larger than two independent binomial random variables
$\bin( \lfloor n p/12 \rfloor, p)$.

Next, we prove that with high probability $Y^+_i\wedge Y^-_i \ge n p /12$. 
By recalling, from \eqref{n/4left}, that 
we have at least $n/4$ vertices left for the pairing algorithm, 
the variable ${Y}_j^+$ is bounded below by
$$
Z_j \sim {\rm Bin}\bR{ \left\lfloor\frac{1}{4}n
\right\rfloor, p} \wedge 
m(n,p).
$$

Hence we have
$$
\begin{aligned}
&\P{{Y}^+_j  \le \dfrac{1}{12} n p\;|\; U_i } \le  \P{Z_j  \le \dfrac{1}{12} n p}.
\end{aligned}
$$
Using the Chernoff bound, we conclude that 
\begin{align*}
&\P{{Y}^+_j \wedge {Y}^-_j \le \dfrac1{12} np \;|\; U_i } \le  e ^{- \gamma np}
\end{align*}
for some constant $\gamma >0$.
Define
\begin{align}
\label{Qi}
Q_i := \bigcap_{j=i}^{i+\floor{d(n,p)}} \bC{Y^+_j \wedge Y^-_j \ge n p/12}.
\end{align}
Using the union bound, we obtain
\begin{align}
\label{probQi}
\P{Q_i^c \mid  U_i } 
&= 
\P{ \bigcup_{j=i}^{i+\floor{d(n,p)}} \bC{Y^+_j \wedge Y^-_j \le n p/12} \mid U_i } \notag\\
&\le
(d(n,p)+1) 
\P{{Y}^+_j \wedge {Y}^-_j \le \dfrac1{12} np \;|\; U_i }
= o(1).
\end{align}




The conditional probability (under $\mathbf{P}_p$) that $r_i$ is paired to $r_j$, given that
$r_j$ is not deleted by the time $\tau_i$ and event $Q_i \cap U_i$, is at least 
\begin{align}
\label{zeta}
\bR{ 1 - (1- p)^{\frac{1}{12} n p} }^2
(1- o(1)) + o(1) 
\ge 
\dfrac{1}{150}
n^2p^4 \bR{ 1 - \dfrac{
np^2}{12-np^2} }^2 =: \zeta,
\end{align}
for all sufficiently large $n$,
by noting that
 $np^2 \le  e^{-104} 2^{-16}$
and
\begin{align*}
1 - (1- p)^{\frac{1}{12} n p}
\ge \dfrac{1}{12}np^2 - \sum_{i \ge 2} 
\binom{\lceil\frac{1}{12}np\rceil}{i} p^i
\ge \dfrac{1}{12}np^2 
\bR{ 1 - \sum_{i \ge 1}
\bR{ \frac{1}{12}np^2 }^ i }.
\end{align*}

Recall that 
when processing vertex $r_i$, we have at least
\begin{align*}
\floor{ d(n, p)/2 }
\ge 2^{18} \useconstant{para} 
\end{align*}
available vertices $r_j$ within distance $d(n,p)$ as potential candidates for pairing with vertex $r_i$.
Then we have, using \eqref{probQi} and \eqref{Qi}, 
that
\begin{align}
\label{eq:bouneta}
\P{r_i \text{ is paired}\;|\; U_i }
&= \P{r_i \text{ is paired} \mid Q_i \cap U_i }
\P{ Q_i  \mid U_i } 
\notag\\
&\qquad\qquad\qquad\qquad
+ \P{r_i \text{ is paired} \mid Q_i^c \cap U_i }
\P{ Q_i^c  \mid   U_i } \notag\\
&\ge
(1- (1-\zeta)^{ 2^{18} \useconstant{para} }) (1-   {o(1)}) + o(1).
\end{align}

In view of the definition of 
$\useconstant{para}$ in \eqref{k2def} and $\zeta$ defined in \eqref{zeta}, we have
\begin{align*}
(1-\zeta)^{ 2^{18} \useconstant{para} }
\le 
\Exp{- 2^{18} \zeta\useconstant{para} }
&\le 
\Exp{-
\bR{ 1 - \dfrac{
np^2}{12-np^2} }^2 2^{10} }
\le 0.00001.
\end{align*}
By plugging in \eqref{eq:bouneta}, we obtain \eqref{eq:pivin}
and complete the proof.
\end{proof}


\subsection{An efficient algorithm for Regime 2}

We can now prove Theorem~\ref{thm-local}, that our generic algorithm gives an efficient algorithm for Regime 2.




\begin{proof}[Proof of Theorem \ref{thm-local}]
Similarly to Lemma \ref{Pfail}, we have
\begin{align*}
&\P{ \mathcal{D}_n \text{ has no
stable $4$-cycle and no PNE}} \\
&\le 
\P{\D \mbox{ has less than $n/100$ disjoint $d(n, p)$-localized $4$-cycles and no PNE}} \\
&\qquad+ \P{ 
\D \text{ has at least $n/100$ disjoint unstable $d(n, p)$-localized pairs}  } \\
&\le 
\P{\mathcal{Z}\mbox{ has less than $n/100$ disjoint $d(n, p)$-localized $4$-cycles and no PNE}} \\
&\qquad+ \P{ 
\D \text{ has at least $n/100$ disjoint unstable $d(n, p)$-localized pairs}  }.
\end{align*}
Then, by \eqref{time0},
the expected running time of the generic algorithm (used with $\ell=2$) is at most
\begin{align*}
O(n^2) + \O{ e^{-p^2n^2}  \binom{n}{2}^2
\poly(n) } + \O{ \clh^n }
\P{ \mathcal{D}_n \text{ contains no 
stable $4$-cycle and no PNE}}.
\end{align*}

By Claim \ref{pairing},
\begin{align*}
 \clh^n 
 \P{\mathcal{Z}\mbox{ has less than $n/100$ disjoint $d(n, p)$-localized $4$-cycles and no PNE}}
\le \clh^n e^{-2n}
< e^{-n}.
\end{align*}
From \eqref{PallBad}, we have 
\begin{align*}
 \clh^n \P{ 
\D \text{ has at least $n/100$ disjoint unstable $d(n, p)$-localized pairs}  }
&\le \clh^n e^{-n}.
\end{align*}
This completes the proof.
\end{proof}

\section{Closing some of the gaps via sprinkling}
\label{sec:gap}

Finally, in this section, we use the {\em sprinkling technique}, described below, to cover two gaps in our previous analysis.
Given $p$ and $p_0$ with $p >p_0$, define $p_1:=
(p - p_0)/(1 - p_0)$. We define a measure $\mathbf{P}$ using a two-step exposure of the edges. We first consider a realization of the random digraph  denoted by $\D^{(0)}$ which, under $\mathbf{P}$ has the same distribution as $\D$ under $\mathbf{P}_{p_0}$. Consider the random graph $ \pi(\D^{(0)})$ obtained from $\D^{(0)}$ by adding directed edges, with the following procedure. 
For each potential arc not present in $\mathcal{D}_n^{(0)}$, add it independently with probability $p_1$.
This procedure is often called \emph{sprinkling}. The distribution of  $ \pi(\D^{(0)})$ under $\mathbf{P}$ is the same as $\D$ under $\mathbf{P}_{p}$.

Fix a cycle $H$ in $\D^{(0)}$, we say that $H$ 
is {\it $\delta$-admissible}
if each vertex in $H$ has out-degree less than $(1+\delta)np_0$ in $\D^{(0)}$.
Define $\mathcal{M}_{\delta}(H)$ to be the event that
$H$ is $\delta$-admissible in $\D^{(0)}$.
In the following subsections, we discuss two cases where we focus on cycles of length 4 and 6, respectively.


\subsection{Case I: $e^{-4/3} n^{-2/3}\le p \le 2^9 n^{-2/3}$}

Let
$p_0 = e^{-4/3} n^{-2/3}$.
We consider the regime 
$p_0 < p \le 2^9 n^{-2/3}.
$
The basic idea of the proof below is that,
with sufficiently high probability, $\mathcal{D}_n^{(0)}$ contains many disjoint stable 4-cycles,
and sprinkling preserves enough of these cycles.
Therefore, there is, with high probability, at least one preserved cycle that remains stable.


\begin{tcolorbox}
\begin{tm}
\label{4sprink}
For
 $e^{-4/3} n^{-2/3}\le p \le 2^9 n^{-2/3}$, the expected running time of the generic algorithm with parameter $\ell=2$ 
 is $\O{ n^2 }$.
 \end{tm}
\end{tcolorbox}

Fix a set $H$ consisting of four vertices, containing exactly two vertices in $R$ and exactly two in $C$. 
For any $\bullet \in \{ \D^{(0)}, \D\}$ denote by $W(H, \bullet)$ the event that vertices in $H$ are part of a stable 4-cycle in $\bullet$.

\begin{tcolorbox}
\begin{lm}
\label{lm:sprink}
Let $H$ be a $\delta$-admissible $4$-cycle with vertices $\{r_1,r_2\} \subset R$ and $\{c_1,c_2\} \subset C$ that is stable in $\D^{(0)}$.
The probability that it becomes unstable in $\D$ is at most $2np_1^2 + 4(1+\delta) n p_0p_1$.  
\end{lm}
\end{tcolorbox}
\begin{proof}
Consider a 4-cycle $H$ that is stable in $\D^{(0)}$. We assume that $H$ is $\delta$-admissible.  For it to become unstable in $\D$, new arcs must be added through sprinkling and create common neighbors for vertices on the same side. We consider two cases. \\
{\bf Case 1}: First, through sprinkling, two vertices on the same side could acquire a new common neighbor. This can happen for the vertices of $H$ that are either in $R$ or in $C$, or both. For each side, we need to add two new arcs to create a common neighbor, each with probability $p_1$, and there are at most $n$ possible vertices that could become the common neighbor. This gives a bound of $2np_1^2$ for the probability of new common neighbors created entirely through sprinkling.\\
{\bf Case 2}: 
We now consider vertices that were already neighbors of some vertices of $H$ in $\D^{(0)}$. 
Through sprinkling, any neighbor of $r_1$ connects to $r_2$ with probability $p_1$. As we assume $H$ to be $\delta$-admissible, 
using the union bound,
and repeating the same argument for $r_2, c_1, c_2$,
gives an upper bound $4 (1+ \delta) n p_0 p_1$ on the probability of creating new common neighbors.

Combining cases 1 and 2 above, we have
\begin{align*}
\mathbf{P}
( W^c(H,\D) \mid \D^{(0)}) I_{W(H,\D^{(0)})\cap \mathcal{M}_\delta(H)}
\le 
2np_1^2 + 4(1+ \delta) p_0 p_1 n
\end{align*}
as required,
where $I_{A}$ 
denotes the indicator of event $A$.
\end{proof}

We next bound the probability 
$\mathbf{P}( \mathcal{M}_{\delta_0}(H) )$,
where $\delta_0=n^{-1/8}$,
that is a given 4-cycle $H$
is $\delta_0$-admissible in $\D^{(0)}$.
By the union bound and Chernoff bound, 
\begin{align}
    \label{eq:ch1}
1 - \mathbf{P}( \mathcal{M}_{\delta_0}(H) )
\le
4\mathbf{P}({\rm Bin}(n, p_0) 
    \geq (1+\delta_0)np_0) 
\leq 4\exp\left(-\dfrac{1}{3}\delta_0^2 np_0 \right)
= e^{- \Omega( n^{1/12} ) }.
\end{align}
We will show that there are a large number of disjoint cycles in $\D^{(0)}$, which, combined with the bound above,
allows us to conclude that a large proportion of them will be $\delta_0$-admissible.

We are now ready to prove Theorem \ref{4sprink}.

\begin{proof}[Proof of Theorem \ref{4sprink}]
The expected running time of the generic algorithm with parameter $2$ is at most
\begin{equation}
\label{sprT}
O(n^2) + \O{ e^{-p^2n^2}  n^4 \poly(n) } + \O{ \clh^n }\,
\mathbf{P} ( \D \text{ has no stable $4$-cycle} ).
\end{equation}

Clearly,
\begin{align}
\label{3prob}
&\mathbf{P}_p\left(\D \text{ has no stable $4$-cycle}\right)
\le
\mathbf{P}\left(
\mathcal{D}_n^{(0)} \text{ has}\le
\dfrac n{100} \mbox{ disjoint stable } 4\mbox{-cycles}\right) \notag\\
&\qquad+\mathbf{P}\left( 
\mathcal{D}_n^{(0)} \text{ has} \ge
\dfrac n{100} \mbox{ disjoint stable } 4\mbox{-cycles}\;
and \;
\D
\text{ has no stable } 4\mbox{-cycle} \right).
\end{align}

We first work on the second probability appearing on the right-hand side of \eqref{3prob}. 
We achieve this in two steps.\\
\noindent {\bf Step 1}. 
Let $\mathcal B$ be the event that 
the number of disjoint stable 4-cycles $H$ in
$\D^{(0)}$ such that $\mathcal{M}_{\delta_0}(H)$ holds is at least $m = \floor{n/101}$. 
On event $\mathcal{B}$,  we define the set of random cycles $\mathcal C:= \{C_1, C_2, \ldots, C_m\}$ by considering the first $m$ disjoint cycles in
$\D^{(0)}$, according to an order fixed a priori. 
Using Lemma~\ref{lm:sprink}, we obtain
\begin{equation}
\begin{aligned}
\label{destroy}
&\mathbf{P}\left( 
\mathcal B
 \cap
\{   \D
\text{ has no disjoint stable } 4\mbox{-cycle}  \} \right)  \\
&\le{\bf E}\left[\sum_{\bm c} \prod_{i=1}^m \mathbf{P}\left(W^c(c_i, \D) \;|\; \D^{(0)}\right)I_{W(c_i,\D^{(0)})\cap \mathcal{M}_{\delta_0}(c_i)}I_{\mathcal{C} =\bm c }\right]  
\notag\\
&
\le \sum_{\bm c} 
\mathbf{P}(\mathcal{C} =\bm c)
\prod_{i=1}^m (2np_1^2 + 4(1+ \delta_0) p_0 p_1 n) 
= \bR{ 2np_1^2 + 4(1+ \delta_0) p_0 p_1 n   }^{ m} 
=  e ^{- \Omega( n \log n ) },
\end{aligned}
\end{equation}
where the summation in the second line is over 
the set containing all $m$ disjoint cycles in the complete bipartite graph $K_{n,n}$,
and we use the independence among disjoint cycles. \\
\noindent {\bf Step 2}. We next argue that, conditioning on the event that there are at least $n/100$ disjoint stable 4-cycles in $\mathcal{D}_n^{(0)}$, with very high probability, there are at least $n/101$ disjoint stable 4-cycles $H$ such that $\mathcal{M}_{\delta_0}(H)$ holds. 

First notice that the probability that a cycle $H$ does not satisfy $\mathcal{M}_{\delta_0}(H)$
is at most $e ^{-C n^{1/12}}$
for some $C>0$.
The probability that among $n/100$ disjoint cycles there are fewer than $m$ ones that  are $\delta_0$-admissible  is bounded by the probability that a binomial 
$Y \sim {\rm Bin}(\floor{n/100}, e ^{-C n^{1/12}})$ 
is larger than $cn$, where $c = 1/100 - 1/101$. 
By Lemma \ref{tail}, we have
$$
\mathbf{P}(Y\ge cn) 
\le
e ^{-
\Omega(n\log n)}.
$$


Next, we bound the first probability appearing on the right-hand side of \eqref{3prob}. Using the union bound gives
\begin{align*}
&\mathbf{P}\left(
\mathcal{D}_n^{(0)} \text{ has} \le
\dfrac n{100} \mbox{ disjoint stable } 4\mbox{-cycles}\right) \\
&\le 
\mathbf{P}\left(\mathcal{D}_n^{(0)} \text{ has} \le \dfrac n4 \text{ disjoint $4$-cycles}\right)
+ \mathbf{P}\left( 
\mathcal{D}_n^{(0)} \text{ has} \ge \dfrac{6}{25} n
\mbox{ disjoint unstable } 4\mbox{-cycles}\right).
\end{align*}

We bound the two probabilities on the right-hand side above.
Noting that 
$
np_0 = e^{-4/3} n^{1/3},
$
by Lemma~\ref{le:bonc}, we have
\begin{align*}
\mathbf{P}\left( 
\mathcal{D}_n^{(0)} \text{ has} \ge \dfrac{6}{25} n
\mbox{ disjoint unstable } 4\mbox{-cycles}\right)
\le
\bR{ \dfrac{50}{3} e n^{4} p^{6}_0 }^{ 96 n/ 400}
< e^{-n}.
\end{align*}
By Claim \ref{cl:X=0},
\begin{align*}
&\mathbf{P}\left(\mathcal{D}_n^{(0)} \text{ has $\le\dfrac n4$ disjoint $4$-cycles}\right)
\\
&\qquad\le 
\dfrac{4^n}{ n} \Exp{
- \dfrac{ 1 }{12 \cdot 2^4} 
\bR{ n p_0 }^{4} }
+ 
\dfrac{4^n}{ n} \Exp{ 
- \dfrac{1}{150 \cdot 2^{8}} n^2 p_0 }
= \O{ e^{-\omega(n)} }.
\end{align*}
This completes the proof
in view of \eqref{sprT} and \eqref{3prob}.
\end{proof}

\subsection{Case II:
$n^{-3/4+ 3/(2\log n)} \le p \le n^{-3/4+ 2/\log n}$}

We set $p_0 = n^{-\frac34 - \frac{3}{2\log n}}$, 
focus on the case $p = n^{-1 + \delta}$,
where
\begin{align*}
\delta \in
\bR{ \dfrac14 - \dfrac{3}{2\log n},
\dfrac14 + \dfrac{2}{\log n} },
\end{align*}
and obtain the following bound on the expected running time,
similar to that in the previous subsection.
\begin{tcolorbox}
\begin{tm}
For
 $n^{-3/4+ 3/(2\log n)} \le p \le n^{-3/4+ 2/\log n}$, the expected running time of the generic algorithm with parameter $\ell=3$ 
 is $\O{ n^2 }$.
 \end{tm}
\end{tcolorbox}

\begin{proof}
Recall that the expected running time of the generic algorithm with parameter $\ell=3$ is at most
\begin{equation}
\label{sprT2}
O(n^2) + \O{ e^{-p^2n^2}  n^6 \poly(n) } + \O{ \clh^n }\,
\mathbf{P} (\D \text{ has no stable $6$-cycle} ).
\end{equation}

Note that
\begin{align}
\label{3prob2}
&\mathbf{P}_p( \D \text{ has no stable $6$-cycle} )
\le
\mathbf{P}\left(
\mathcal{D}_n^{(0)} \text{ has less than }
\dfrac n{100} \mbox{ disjoint stable } 6\mbox{-cycles}\right) \notag\\
&\quad+
\mathbf{P}\left( 
\mathcal{D}_n^{(0)} \text{ has at least }
\dfrac n{100} \mbox{ disjoint stable } 6\mbox{-cycles}
\cap
\D \text{ has no disjoint stable } 6\mbox{-cycle} \right). 
\end{align}

The analysis similar to \eqref{eq:ch1} yields that, for $\delta_0= n^{-1/9}$,
\begin{align*}
1 - \mathbf{P}( \mathcal{M}_{\delta_0}(H) )
\le
6\exp\left(-\dfrac{1}{3}\delta_0^2 np_0 \right)
= e^{- \Omega( n^{1/36} ) }.
\end{align*}
Moreover,
$$
\begin{aligned}
&\mathbf{P}\left( 
\mathcal{D}_n^{(0)} \text{ has at least }
\dfrac n{101} \mbox{ disjoint stable $\delta_0$-admissible } 6\mbox{-cycles}\;
and \;
\D
\text{ has no stable } 6\mbox{-cycle} \right) \\
&\le \bR{ 6np_1^2 + 12 n(1+\delta)p_0p_1 }^{ n/101} =  e ^{-
\Omega(n\log n)}.
\end{aligned}
$$


By the union bound,
\begin{align*}
&\mathbf{P}\left(
\mathcal{D}_n^{(0)} \text{ has less than }
\dfrac n{100} \mbox{ disjoint stable } 6\mbox{-cycles}\right) 
\le 
\mathbf{P} \left(\mathcal{D}_n^{(0)} \text{ has less than $\dfrac n6$ disjoint $6$-cycles} \right) \\
&\qquad+ \mathbf{P}\left( 
\mathcal{D}_n^{(0)} \text{ has at least } \dfrac{47}{300} n
\mbox{ disjoint unstable } 6\mbox{-cycles}\right).
\end{align*}


By Lemma \ref{le:bonc}, we have 
\begin{align*}
\mathbf{P}\left( 
\mathcal{D}_n^{(0)} \text{ has at least } \dfrac{47}{300} n
\mbox{ disjoint unstable } 4\mbox{-cycles}\right)
\le
\bR{\dfrac{1800}{47} e n^{6} p^{8}_0 }^{\frac{47}{300} n},
\end{align*}
By Claim \ref{cl:X=0}, we have
\begin{align*}
&\mathbf{P} \left( \mathcal{D}_n^{(0)} \text{ has less than $\dfrac n6$ disjoint $6$-cycles} \right) \\
&\qquad\le 
\dfrac{4^n}{ n} \Exp{
- \dfrac{ 1 }{36 \cdot 2^6} 
\bR{ n p_0 }^{6} }
+ 
\dfrac{4^n}{ n} \Exp{ 
- \dfrac{1}{150 \cdot 3^8}n^2 p_0 }.
\end{align*}
Finally, the probability that among $n/100$ disjoint cycles there are less than $n/101$  that  are $\delta_0$-admissible is bounded, using the Chernoff bound, by $ e ^{-
\Omega(n\log n)}$.
This completes the proof
in view of \eqref{sprT2} and \eqref{3prob2}.
\end{proof}

\section{Conclusion}\label{sec:conc}
We presented an expected polynomial-time algorithm for random
win-lose games, for nearly all values of $p$. Two important open problems remain. The first is to determine the computational complexity with respect to the two small intervals for $p$ left unresolved in this paper. 
Second, beyond Bernoulli distributions, do expected polynomial-time algorithms exist for random games in which the payoff entries are drawn independently from other distributions, such as the normal and uniform distributions?


%

\begin{thebibliography}{99}

\bibitem{AKV05} T. Abbott, D. Kane and P. Valiant, ``On the complexity 
of two-player win-lose games'', {\em Proceedings of $46$th Symposium on 
Foundations of Computer Science (FOCS)}, pp113--122, 2005.

\bibitem{AOV07} L. Addario-Berry, N. Olver and A. Vetta, ``A polynomial-time algorithm for finding Nash equilibria in 
planar win-lose games'', {\em Journal of Graph Algorithms and Applications}, {\bf 11(1)}, pp309--319, 2007.

\bibitem{AGM11}
B. Adsul, J. Garg, R. Mehta and M. Sohoni, ``Rank-1 bimatrix games: a homeomorphism and a polynomial time algorithm", 
{\em Proceedings of $43$th Symposium on the Theory of Computing (STOC)}, pp195--204, 2011.

\bibitem{ACSZ19}
B. Amiet, A. Collevecchio, M. Scarsini and Z.
 Zhong, ``Pure {N}ash equilibria and best-response dynamics in random
 games",
 {\em Mathematics of Operations Research},
 {\bf 46}, pp1552--1572, 2021.
 
\bibitem{BVV07} I. B\'{a}r\'{a}ny, S. Vempala and A. Vetta, ``Nash equilibria
in random games'', {\em Random Structures and Algorithms}, {\bf 31(4)}, pp391--405, 2007.

\bibitem{CDT06} 
X. Chen, X. Deng and S. Teng, ``Sparse games are hard'', {\em Proceedings of the $2$nd Workshop on Internet and Network Economics (WINE)}, pp262--273, 2006.

\bibitem{CDT09} 
X. Chen, X. Deng and S. Teng, ``Settling the complexity of computing two-player
Nash equilibria'', {\em Journal of the ACM}, {\bf 56(3)}, pp1--57, 2009.


\bibitem{CLR06}
B. Codenotti, M. Leoncini and G. Resta,
``Efficient computation of {N}ash equilibria for very sparse win-lose bimatrix games", {\em Proceedings of the $14$th European Symposium on Algorithms (ESA)}, pp232--243, 2006.


\bibitem{DGP09} C. Daskalakis, P. Goldberg and C. Papadimitriou, ``The complexity of computing a Nash equilibrium", 
{\em SIAM Journal on Computing}, {\bf 39(1)}, pp195--259, 2009.

\bibitem{Gold57}
A. Goldman, ``The probability of a saddlepoint'', {\em The American Mathematical Monthly}, {\bf 64(10)}, pp729--730, 1957.

\bibitem{GGN68}
K. Goldberg, A. Goldman and M. Newman, ``The probability of an equilibrium point'', {\em Journal of Research of the National Bureau of Standards}, {\bf 72B(2)}, pp93--101, 1968.


\bibitem{JLR00}
S.~Janson, T.~{\L}uczak and A.~Ruci\'nski, {\em Random Graphs}, 
John Wiley \& Sons, 2000.

\bibitem{LH64}
C. Lemke and J. Howson, ``Equilibrium points in bimatrix games'', {\em Journal of the Society for Industrial and Applied Mathematics}, {\bf 12}, pp413--423, 1964.


\bibitem{McL05}
A. McLennan, ``The expected number of Nash equilibria of a normal form game'',
{\em Econometrica}, {\bf 73(1)}, pp141--174, 2005.

\bibitem{MB05}
A. McLennan and J. Berg, ``Asymptotic expected number of Nash equilibria of a two-player normal form games'',
{\em Games and Economic Behavior}, {\bf 51(2)}, pp264--295, 2005.


\bibitem{Meh14}
R. Mehta, ``Constant rank two-player games are PPAD-hard'',
{\em Proceedings of 46th Symposium on the Theory of Computing (STOC)}, pp258--267, 2014.

\bibitem{Nash51}
J. Nash, ``Non-cooperative games'',
{\em Annals of Mathematics}, {\bf 54(2)}, pp286--295, 1951.

\bibitem{Pap01}
C. Papadimitriou, ``Algorithms, games, and the Internet'',
{\em Proceedings of the 33rd ACM Symposium on Theory of Computing (STOC)}, pp49–-753, 2001.

\bibitem{R2000}
Y. Rinott and M. Scarsini, ``On the number of pure strategy Nash equilibria in random games", 
{\em Games and Economic Behavior}, {\bf 33(2)}, pp274--293, 2000.
 

\bibitem{savani2006hard}
R. Savani and B. Von~Stengel,  ``Hard-to-solve bimatrix games'',
\newblock {\em Econometrica}, {\bf 74(2)}, pp397--429, 2006.


\bibitem{VN28}
J. Von~Neumann, ``Zur Theorie der Gesellschaftsspiele'', {\em Mathematische Annalen},
{\bf 100}, pp295–-320, 1928.

\bibitem{von2002computing}
B. Von~Stengel, ``Computing equilibria for two-person games'',
{\em Handbook of Game Theory with Economic Applications},
 Vol. 3, pp1723--1759, 2002.



\end{thebibliography}
\end{document}